\documentclass[11pt]{article}
\usepackage[margin=1in]{geometry}
\usepackage{times}
\usepackage[compact]{titlesec}
\usepackage{amsfonts}
\usepackage{amsmath}
\usepackage{amsthm}
\usepackage{amssymb}
\usepackage{latexsym}
\usepackage{epic}
\usepackage{epsfig}
\usepackage{hyperref}
\usepackage{verbatim}
\usepackage[justification=centering]{caption}
\usepackage{enumitem}
\usepackage{color}
\usepackage{xcolor}
\allowdisplaybreaks[1]
\usepackage[numbers]{natbib}
\usepackage{setspace}
\usepackage{pdfpages}
\usepackage{multicol}
\usepackage{graphicx}
\usepackage{subcaption}
\usepackage{mathtools} 
\usepackage[ruled, linesnumbered]{algorithm2e} 

\usepackage{fullpage}

\usepackage{float,framed}
 
\definecolor{darkspringgreen}{rgb}{0.09, 0.45, 0.27}

\title{Auctioning with Strategically Reticent Bidders\thanks{The work is supported by a Google Research Faculty Award and NSF grant CCF-2303372. }} 

\author{Jibang Wu\\  University of Chicago \\ \small wujibang@uchicago.edu 
        \and 
        Ashwinkumar Badanidiyuru \\ Google  \\ \small  ashwinkumarbv@google.com 
        \and 
        Haifeng Xu\\  University of Chicago \\ \small haifengxu@uchicago.edu 
        }

\newtheoremstyle{nonindented}{1ex}{1ex}{}{}{\bfseries}{.}{.5em}{}
\newtheoremstyle{indented}{1ex}{1ex}{\itshape\addtolength{\leftskip}{0.6cm}\addtolength{\rightskip}{0.6cm}}{}{\bfseries}{.}{.5em}{}
\theoremstyle{nonindented}
\theoremstyle{indented}
\theoremstyle{plain}
 
\newtheorem{theorem}{Theorem}[section]
\newtheorem{corollary}[theorem]{Corollary}

\newtheorem{proposition}[theorem]{Proposition}

\newtheorem{definition}[theorem]{Definition}
\newtheorem{remark}[theorem]{Remark}

\newtheorem{example}[theorem]{Example}

\newtheorem*{conjecture*}{Conjecture}

\renewcommand{\bar}{\overline}

\DeclareMathOperator{\mps}{\geq_{\text{MPS}}}

\def\max2{\qopname\relax n{max2}}
\def\max{\qopname\relax n{max}}

\def\argmax{\qopname\relax n{argmax}}

\def\Pr{\qopname\relax n{\mathbf{Pr}}}
\def\Ex{\qopname\relax n{\mathbf{E}}}
\def\Var{\qopname\relax n{\mathbf{Var}}}
\def\Cov{\qopname\relax n{\mathbf{Cov}}}

\newcommand{\RR}{\mathbb{R}}

\newcommand{\bsigma}{\boldsymbol{\sigma}}

\def\G{\mathcal{G}}

\def\M{\mathcal{M}}

\def\S{\mathcal{S}}

\def\s{\boldsymbol{s}}
\def\states{ {\boldsymbol \theta} }
\def\types{ \boldsymbol{t}}
\def\bids{\boldsymbol{b}} 
\def\part{P} 
 
\def\v{\boldsymbol{v}}

\def\bsigma{\boldsymbol{\sigma}}
\def\blambda{\boldsymbol{\lambda}}
\def\btau{\boldsymbol{\tau}}
\def\util{\texttt{U}_{i}}
\def\exputil{\bar{\texttt{U}}_{i}}

\newcommand{\indep}{\perp \!\!\! \perp}

\newenvironment{lp*}{\begin{equation*}  \begin{array}{lll}}{\end{array}\end{equation*}}

\newcommand{\oldmodel}{\texttt{Mis-Bidders}} 
\newcommand{\newmodel}{\texttt{SR-Bidders}} 

\newcommand{\infobid}{\texttt{SR-bidder}} 
\newcommand{\infoIC}{{IIC}}

\newcommand{\emm}{{Expected Meta Mechanism}} 
\newcommand{\smm}{{Simulated Meta Mechanism}}

\begin{document}


\date{}

\maketitle

\begin{abstract}

We propose and study a novel mechanism design setup where each bidder holds two kinds of private information: (1)   \emph{value type}, which can be misreported; (2)   \emph{information variable}, which the bidder may  want to conceal or partially reveal, but importantly, \emph{not} to misreport.  
We refer to bidders with such behaviors  as \emph{strategically reticent bidders}. Among others, one direct motivation of our model is the ad auction in which many ad platforms today elicit from each bidder not only their private value per conversion but also their private information about Internet users (e.g., user activities on the advertiser's websites) in order to improve the platform's estimation of conversion rates.     

We show that in this new setup, it is still possible to design mechanisms that are both \emph{Incentive and Information Compatible} (IIC). We develop  two different black-box transformations, which convert any mechanism $\mathcal{M}$ for classic bidders to a mechanism $\bar{\mathcal{M}}$ for strategically reticent bidders,  based on either outcome of expectation or expectation of outcome, respectively. We identify  properties of the original mechanism $\mathcal{M}$ under which the transformation leads to IIC mechanisms $\bar{\mathcal{M}}$. Interestingly, as corollaries of these results, we show  that running VCG with bidders' \emph{expected values} maximizes welfare, whereas the mechanism using \emph{expected outcome} of Myerson's auction maximizes revenue. Finally, we study how regulation on the auctioneer's usage of information can lead to more robust mechanisms.

\end{abstract}
 

\section{Introduction}

In standard auction models with private values (independent or correlated),  bidders are often assumed to  know their values of the auctioned items. Therefore,  the mechanism designer will undertake a \emph{value elicitation} procedure, during which each bidder is asked to report his value (i.e., the \emph{bid}). 
However, in many applications of today's data-driven world, it becomes more and more difficult for bidders to know their exact values of the to-be-sold items. For example, in online ad auctions, advertisers, as the bidders, usually do not have sufficient data or computation resources to estimate their exact values for an ad slot, e.g., the conversion values of an Internet user.
Instead, each bidder may be able to access some information, such as user activities on his store websites, that are predictive of the conversion values. 
In this case, the auctioneers cannot elicit bidders' exact values anymore (as bidders themselves do not even know their values), but they can go through an \emph{information elicitation} procedure to incentivize bidders to share their \emph{information variable}.\footnote{In interdependent value auctions \cite{milgrom1982theory}, the bidder's information is also called a \emph{signal}. We shall call it an information variable instead, as we will reserve the term ``signal'' for other use, specifically, as the carrier of partial information about the information variable. It will also be differentiated from the term ``signaling scheme'' as a randomized procedure that implements partial information revelation.} Classic works, e.g., models of interdependent value auctions \cite{milgrom1982theory,roughgarden2013optimal}, have treated such information elicitation almost the same as value elicitation --- the bidders are assumed to either truthfully reveal the information variable they observed or misreport any other possible variable from its support. 
While this could be a valid assumption, we argue below that there are many situations in which \emph{partial information revelation}, as opposed to information \emph{misreporting}, can better model the bidders' strategic behaviors, especially under the new norms of information economics today.
    

Firstly, in some applications, partial information revelation is already a more common practice than misreporting information. For instance, today's advertisers are used to installing web-scripts or sharing cookies access on their websites to Internet ad platforms for conversion tracking~\cite{rosales2012post,CVRtracking,JoinData}. Any advertiser as a bidder can decide what type of information will be revealed to the platform (e.g., purchase behaviors, browsing activities, or some mixture) and how fine-grained the revealed information is (ranging from detailed user browsing trajectories to monthly user activity statistics). These choices naturally correspond to the behaviors of selective information revelation. In contrast, it is almost impossible for a bidder to misreport such information, as generating fake user cookies often involves highly technical and costly procedures.  

Secondly, the ad auction application motivates another key reason to model bidders' partial information revelation: the behavior of misreporting information may have moral or even legal consequences for advertisers. On one hand, regulations  such as the European Commission's \emph{Code of Practice on Disinformation} have already been imposed to prevent disinformation on Internet platforms. On the other hand, many advertisers are large cooperations with good business ethics and   high reputations; they tend to avoid any disinformation in their marketing practices.  Nevertheless, it is a perfectly reasonable business strategy to reveal partial information, i.e., determine what types of information and to what levels they want ad exchange platforms to track.

Thirdly, taking a step back, even if misreporting information could be a valid option for bidders, revealing partial information could lead to strictly more utility, purely from the incentive perspective. In fact, the classic bidder behavior model  cannot capture the simple situation where the bidder chooses to reveal no information about the information variable he observed. The following Example \ref{ex:silent-power} shows that, from a bidder's perspective, the option of staying silent about his private information can be much better than being forced to  report any state.


\begin{example}[The Incentive for  Silence] \label{ex:silent-power}  
Consider a single-item auction with 3 bidders. Only bidder~1 has a private information variable $\theta_1 \in \{ 0,1 \}$, called a \emph{state}, which is a uniformly random bit. The state  determines every bidder's value:  $v_1(\theta_1) = 0.9, v_2(\theta_1) = \theta_1, v_3(\theta_1) = 1 - \theta_1$. Bidder $1$ can privately observe the realized $\theta_1$ but all other bidders and the auctioneer only know its distribution. 

Suppose the auctioneer  runs the standard second price auction. If  bidder 1 is forced to report a state (truthfully or untruthfully), his utility will always be $0$ since only bidder 2 or 3 can be the winner, depending on $\theta_1$ is $1$ or $0$.  
In contrast, suppose the  auction is modified to allow an extra action for  bidders -- ``silence'': when a bidder chooses to be silent, the auctioneer will use the expected bidder values to run the second price auction. In this case, if bidder 1 simply keeps silent, he will always win and receive utility $0.4=0.9 - 0.5$, because the values of bidder 2 and 3 are both estimated as $0.5$.  
\end{example}

\subsection{Our Contributions } 
{The key conceptual contribution of this paper is the} formulation of a novel class of auction design problem with strategically reticent bidders, conveniently referred to as \newmodel{}.\footnote{The term ``strategic reticence'' is coined by \citet{chen2007bluffing} to characterize trader's information hiding behavior in prediction markets. Our paper adopts this term to describe the bidder's partial revelation behavior in auctions.} {Assembling the different primitives together into a clean model turns out to be a highly non-trivial task, and we elaborate on a few key challenges in Section \ref{sec:model-remark}.}  In our model,  each bidder privately observes a \emph{type variable} and an \emph{information variable} that jointly determine his value. They capture different types of user knowledge --- while the type variable may be misreported, the information variable can be partially revealed but cannot be misreported.
To design mechanisms that take both sources of information as input and accommodate the different strategic behavior patterns, we adopt the solution concepts of classic mechanism design and introduce {an augmented truthfulness notion of} \emph{Incentive and Information Compatibility} (\infoIC{}), under which all bidders truthfully report their type variables and also fully reveal their information variables.
We show that designing truthful mechanisms for \newmodel{} is still without loss of generality by invoking the  revelation principle.   

Mechanism design has already established a thorough understanding of many classic problems thus we do not plan to ``reinvent all the wheels'' for our new setup.
Instead, we employ a fundamental paradigm from algorithm design and use  \emph{reduction} to build   black-box transformations, coined \emph{meta mechanisms}, that can convert any classic mechanism to a mechanism for our generalized bidder behavior model. For two natural classes of meta mechanisms, based on either the ``outcome of the expectations'' or the ``expectation of the outcomes'', we identify conditions on the original mechanisms that will make the new mechanism IIC for \newmodel{}. Interestingly, as corollaries of these two meta mechanisms, we show that the meta mechanism simply using the outcome of Vickrey auction executed with expected bidder values achieves maximum welfare. In contrast, a different meta mechanism that outputs the expectation of the outcomes generated by Myerson's  auction for each possible bidder value profile maximizes the revenue under \newmodel{}. 
One caveat about both results above is that both mechanisms can only guarantee ex-post \infoIC{} --- unlike the dominant-strategy IC of Vickrey and Myerson's auction for classic setups. We demonstrate that this relaxation is inevitable since it is impossible to achieve dominant-strategy IIC in our setup. 

Finally, we consider the situation in which the auctioneer can further commit to never estimating a bidder $i$'s value with other bidders' information (despite others' information can be used to refine the value estimation of bidder $i$ due to correlation). In reality, such commitment may be enforced by regulation policies that limit the usage of a bidder's private information.  Perhaps surprisingly, it turns out that such restriction strictly benefits the auctioneer --- it restores the dominant strategy \infoIC{} in both mechanisms while preserving the auctioneer's objective value. 

\subsection{Discussions on Related Works}

Information disclosure in auctions has been extensively studied. However, most previous works  focused on the opposite side of ours, i.e., information revelation by the auctioneer. Along that line, \citet{milgrom1982theory} lay down the foundation and propose the well-known ``Linkage Principle'' in a symmetric interdependent-value auction setup and show that the auctioneer's revenue will also improve by  revealing more information about the item to bidders. Nevertheless, it is later shown that such full transparency may not be optimal in other settings, e.g., when the bidder values are asymmetric or when the auction formats change~\cite{li2019revenue, hausch1987asymmetric, perry1999failure,Emek12, bro2012send,Syrgkanis13,Badanidiyuru2018targeting}.
Specifically, the auctioneer may increase revenue by carefully obfuscating   item information through signaling schemes. \citet{fu2012ad} show that, under mild assumptions, the auctioneer  maximizes the revenue in the Myerson's optimal auction by  revealing full information to the bidder. In another related work, \citet{daskalakis2016does} consider the joint design of optimal signaling and auctioning mechanism. 


Different from all works mentioned above, this paper focuses on the setup where bidders have the information advantage and may strategically reveal it to the auctioneer. 
As far as we know, there have been very few works concerning such an \newmodel{} behavior model. The closest to our work is perhaps that of \citet{shen2019buyer}, which considers bidder signaling in the private value auction model to influence the auctioneer's \emph{prior-dependent} mechanisms and characterize bidders' signaling equilibrium. However, their model and research questions are both different from ours. Specifically, our model has both type variables, which the bidder can misreport, and information variables, which the bidder can partially reveal. In contrast, bidders in the model of \citet{shen2019buyer} only have private value types and may signal information about the type. More importantly, our research focuses on the auctioneer's design problem of optimal mechanism, whereas \citet{shen2019buyer} fix the auction format and study bidders' design problem of optimal signaling. 

Another related line of works considers  ``signaling bidder'' \cite{bos2020optimal}, a notion to capture the bidder behavior model with signaling concern, that is, how the bids are perceived by the outsider in the aftermarket -- such incentives  exist particularly in artwork \cite{mandel2009art} or takeover \cite{liu2012takeover} bidding to establish bidder's reputation~\cite{giovannoni2014reputational, wan2009rfq}. So in these situations, the bidder's utility is not only decided by the allocation and payment but also affected by the posterior belief of his value. In contrast, the \newmodel{} only care about how their bids are perceived by the auctioneer and its consequential effect on payment and allocation.   


Finally, there have also been studies about an information intermediary who has the knowledge of  bidders' values and signals this information to the auctioneer so as to influence the outcome of a mechanism, akin to Bayesian persuasion~\cite{Kamenica2011}. The seminal work by \cite{bergemann2015limits} provides a clean characterization of the feasible set of induced consumer surplus and seller revenue in the bilateral trading setup. Several works~\cite{shen2018closed, cummings2020algorithmic, cai2020third,alijani2020limits} have followed up to extend this result to more general or slightly modified settings. 




\section{Preliminaries on Information Revelation}\label{sec:prelim}

Since information revelation is central to our study, it is worthwhile to formally introduce the standard models for describing information revelation, i.e., signaling schemes and the induced Blackwell experiments.  Consider a random variable $\theta$ drawn from a commonly known prior distribution $g$ supported on  a finite state space $\Theta$, i.e. $g  \in \Delta(\Theta)$.
A \emph{signaling scheme} is a pair $(\pi, \Sigma)$ where $\pi: \Theta \to \Delta(\Sigma)$ and $\Sigma$ is the support of signals. $\pi(\sigma | \theta)$ specifies the probability of sending signal $\sigma \in \Sigma$ conditioned on realizing state $\theta \in \Theta$. The signal $\sigma$  carries partial information about the state $\theta$ via its correlation with $\theta$.  Formally, with the prior belief $g$,  a receiver of   $\sigma$  can derive an updated posterior distribution about $\theta$  via the Bayes' rule: $\Pr(\theta|\sigma)= \frac{ \pi(\sigma | \theta) g(\theta) }{ \sum_{\theta\in \Theta} \pi(\sigma | \theta) g(\theta) }$. This means one can equivalently think of a signal $\sigma$  as a distribution over state space $\Theta$, i.e., $\sigma \in \Delta(\Theta)$. Therefore, a signaling scheme effectively induces a distribution over all possible beliefs in $\Delta(\Theta)$, which is realized to $\sigma \in \Delta(\Theta)$ with probability $\Pr(\sigma) = \sum_{\theta \in \Theta} \pi(\sigma|\theta)g(\theta) $. 
This leads to the notion of \emph{(Blackwell) experiment}  \cite{blackwell1953equivalent}, which is  described by $\tau \in \Delta(\Delta(\Theta))$ as a distribution of beliefs satisfying the law of total probability mass $\Ex_{\sigma \sim \tau}[\sigma]=g$, where the experiment outcome $\sigma \sim \tau$ is directly a belief of state.
We remark that, from the perspective of the information receiver, there is a subtle difference between the two notions: a signaling scheme spells out the signal generation process $\pi$, whereas an experiment only describes a distribution over beliefs~\cite{brooks2022information}. For the purpose of analysis in this paper, these interpretations of partial revelation are equivalent, though their subtle difference matters from a modeling perspective --- we defer the detailed discussion to Section~\ref{sec:model-remark}. 

A seminal result due to \citet{blackwell1953equivalent} establishes how to compare the informativeness of two experiments $\tau, \lambda$. We say an experiment $\tau$ is more informative than $\lambda$, if there exists a signaling scheme that maps a signal $\sigma \sim \tau$ to a signal $s \sim \lambda$, or equivalently, $\tau$ is a \emph{mean preserving spread} of $\lambda$.  By definition, an experiment $\tau$ is a mean preserving spread of another experiment $\lambda$, denoted as $\tau \mps \lambda$, if the experiment outcomes $\sigma \sim \tau, s \sim \lambda$ can be correlated so that $\Ex[\sigma | s] = s$ for any $s$. 
If so,  a signaling scheme $\pi: \Sigma \to \Delta(\S)$ that generates experiment $\lambda$ can be constructed based on their correlation, i.e., $\pi(s|\sigma)=\Pr(s|\sigma)$. 
Such partial order of experiments is also known as the second order stochastic dominance, or Blackwell informativeness under different contexts.
Moreover, the informativeness of experiments closely connects to decision problems. Our result  will need the following useful and well-known result  of \citet{blackwell1953equivalent}. 
\begin{theorem}[\citet{blackwell1953equivalent}] \label{thm:blackwell}
$\tau \in \Delta(\Delta(\Theta)) $ is a mean-preserving spread of $ \lambda\in \Delta(\Delta(\Theta))$ if and only if $\Ex_{\sigma \sim \tau}[u(\sigma)] \geq \Ex_{s \sim \lambda}[u(s)] $ for any convex function $u: \Delta(\Theta) \to \RR$.
\end{theorem} 




We remark that, in practice, the auctioneer's predictive model based on the bidder's feature set can be viewed as a Blackwell experiment; as long as the label set is public knowledge (e.g., whether the user converted is traced by auctioneer's web-scripts), any \emph{consistent} modification (such as noise injection, attribute deletion, re-scaling and grouping) on the feature revealed to the auctioneer is mathematically equivalent to an information signaling scheme (i.e., garbling~\cite{marschak1968economic}). 


\begin{remark}
[Connections between feature manipulation and partial information revelation]
Consider a typical classification problem on the dataset $(\Sigma, \Theta)$, where $\sigma \in \Sigma$ is the feature vector of an Internet user, $\theta \in \Theta$ is the user label that relates to the bidder's value. The bidder privately collects a set of features $\Sigma$, while the corresponding set of labels $\Theta$ is public knowledge. If the bidder directly shares the raw features $\Sigma$ with auctioneer, the risk-minimizing classifier learned from this dataset, assuming it is sufficiently large, tends to be the well-known \emph{Bayes optimal classifier}, $\Pr(\theta | \sigma), \forall \theta, \sigma$ \cite{shalev2014understanding}. This Bayes optimal classifier essentially recovers the data generation process: each feature vector $\sigma$ can be viewed as a signal of some Blackwell experiment $\tau$, from which one can infer a distribution of possible user labels $\theta$; the distribution of $\sigma$ in the dataset specifies the signal distribution of $\tau$. 
Meanwhile, if the bidder follows any feature manipulation scheme on $\Sigma$, it can be fully captured by a (possible randomized) mapping $\pi: \Sigma \to \Delta(\S)$ to some transformed feature set $\S$. The best classifier to learn from this manipulated dataset $(\S, \Theta)$ is another Bayes optimal classifier, $\Pr(\theta | s), \forall \theta, s$. The Bayes plausibility implies that $\sum_{s\in \S} \Pr(\theta | s) \pi(s | \sigma) = \Pr(\theta | \sigma), \forall \theta, \sigma$. So $\pi$ can be viewed as a signaling scheme that acts on any realized data signal $\sigma$ to generate a garbled data signal $s$, and the auctioneer learns the posterior $\Pr(\theta | s)$ as her knowledge on the user label. As such, if the auctioneer trains  on sufficiently many data points with garbled feature set $\S$ supplied by a bidder, the auctioneer's Bayes classifier will become an experiment $\lambda$ with $\tau\mps \lambda$.
\end{remark}

\section{A Model of Auctioning with \newmodel{}} \label{sec:model}
In this section, we formally introduce the auction setup with \newmodel, where each bidder has two kinds of private information --- a \emph{type variable} and an \emph{information variable}.  
Like classic auctions, the type may be misreported. So \newmodel{} is a strict generalization of the classic misreporting bidder behavior models, conveniently referred to as \oldmodel{}. What makes \newmodel{} strictly richer than \oldmodel{} is the additional component of information variables that cannot be misreported but only \emph{partially revealed}; this naturally captures the consequence of bidders' strategic behaviors in sharing  data with auctioneer. 
To focus on the strategic aspects of the problem, we will abstract away the machine learning procedures and assume {the dependence of the learning outcome on the bidder's information}
is described by a Blackwell experiment. Employing standard Bayesian inference, the auctioneer will infer a posterior belief about the information variable, conditioned on any realized experiment outcome. We elaborate on our problem next. 

\subsection{A Generalized Value Model}
\vspace{1mm}
\noindent \textbf{Basic Setup. } 
We focus on the basic setup of auction design with single item allocation and correlated bidder values (the insistence of correlated value is due to realistic motivations that an advertiser's data tend to be correlated with others'). 
The item is of interest to $n$ {bidders}. The \emph{exact} value of each bidder $i\in [n]$, denoted as $v_i(t_i; \theta_i)$, is a function of his own type and value-relevant state, $v_i: T_i \times \Theta_i \to \RR$, where $T_i$ is the {possibly infinite} set of types, and $\Theta_i$ is the finite set of states. 
Each bidder $i$ privately observes a realization of \emph{type variable} $t_i \sim h_i$ and \emph{information variable} $\sigma_i \sim \tau_i$, which is a signal from the experiment $\tau_i$. 
Here, the type distribution $h_i$ supports on $T_i$, whereas the information variable itself $\sigma_i \in \Delta(\Theta_i)$ is a belief of state with a distribution that supports on $\Theta_i$. Notably, bidder $i$  does not necessarily observe the underlying state $\theta_i \sim g_i$; this captures the common situations where advertisers' information may only be partially predictive about an Internet user's conversion.
The experiment $\tau_i \in \Delta(\Delta(\Theta_i))$ models the data generation process that the signal $\sigma_i\in \Delta(\Theta_i)$ observable by bidder $i$ (e.g., inferred from an Internet user's browsing behavior) is generated from the underlying value-relevant state $\theta_i$ (e.g., the user's hidden true willingness to purchase). In Section~\ref{sec:model-remark}, we will further discuss on why it is both natural and crucial to model $\theta_i$ instead of $\sigma_i$ as the state in this problem. 
Let $g$ denote the joint prior distribution of $\states$ with each $g_i$ being the marginal distribution of $\theta_i$ satisfies $g_i(\theta_i) = \sum_{\theta_{-i}} g(\theta_i, \theta_{-i})$. We assume the distribution $h, g$ are public knowledge.
Throughout this paper, we make the following assumptions about the information structure of the variables introduced above:
\begin{enumerate}
    \item We assume bidders' type variables $\{t_i\}_{i \in [n]}$ are independent with any of the information variables $\{\sigma_i\}_{i \in [n]}$, since they capture different aspects of the bidders' value (e.g., an advertiser's value per conversion \emph{v.s.} his information of whether a user visits his website). 
    \item We assume bidders' type variables $\{t_i\}_{i \in [n]}$ are independent with each other, like the classic independent type model~\cite{myerson1981optimal}.
    The state variables $\{ \theta_i \}_{i=1}^n $ may be correlated with each other. When there is uncertainty about the state $\theta_i$, the auctioneer may use the knowledge of other state $\theta_j$ and their correlations  to  better estimate $v_i$. 
    \item We assume the conditional independence, $\theta_i \indep \sigma_{j} | \sigma_i, \forall i, j\in [n]$, such that, if $\sigma_i$ is fully known, then $\theta_i$ (and its corresponding value $v_i(t_i; \theta_i)$) becomes independent of any other $\sigma_j$; in other words, $\sigma_i$ is a \emph{sufficient predictor}~\cite{luo2014efficient} of $v_i$.\footnote{Without this assumption, the bidders' values will become \emph{interdependent} since a bidder's value depends on other bidders information even given all his information. We leave this intriguing generalization for future research, and choose to focus on a more basic setup in this work.} Therefore, the value model is \emph{not} interdependent under full information revelation, but becomes interdependent if bidder $i$ reveals partial information, in which case other bidders' signals can refine estimation of $\theta_i$. 
\end{enumerate}



\vspace{1mm}
\noindent \textbf{The Strategically Reticent Bidder Behavior Model  (\newmodel). } 
The main difference between our model and classic auction setups lies in the bidder's behavior. 
Formally, each bidder's action in our model consists of two components: (1) a (possibly untruthful) report $b_i \in T_i$ of his type; (2) a (possibly partially revealing) signal $s_i \in \Delta(\Theta_i)$ of his information. 
Type reporting requires no further explanation. By convention, we also call $b_i$ the \emph{bid} which may be different from the true $t_i$ since misreporting is possible. The bidder's partial revelation of information is modeled as another experiment $\lambda_i$. $s_i \sim \lambda_i$ is the realized outcome of the experiment $\lambda_i \in \Delta(\Delta(\Theta_i))$. Notably, experiment $\lambda_i$ is conducted on top of  $\tau_i$ and thus reveals no more information than $\tau_i$, i.e.,  $\tau_i \mps \lambda_i$. As mentioned, such experiment $\lambda_i$ can be without loss of generality implemented from experiment $\tau_i$ and a signaling scheme $\pi_i: \Sigma_i \to \Delta(\S_i)$ that specifies the correlation between $\sigma_i \sim \tau_i$ and $s_i\sim \lambda_i$. We will think of the bidder $i$ directly commits to an experiment $\lambda_i$ (induced by some $\pi_i$ and $\lambda_i$) and the signals $s_i \sim \lambda_i$ are drawn from the experiment. This is in fact necessary to avoid degeneracy in its mechanism design problem, since the auctioneer in this case need not to know $\tau_i$ (see detailed discussion in Section \ref{sec:model-remark}).
We will use the notation $\tau_i(\sigma_i), \lambda_i(s_i)$ to denote the probability of the corresponding signal realization from the experiments.


We adopt the standard assumption in the literature of information elicitation and information design \cite{wolinsky2002eliciting, dughmi2017algorithmic, dughmi2019algorithmic,  bergemann2019information,kamenica2019bayesian}, and assume each bidder $i$ commits to an {experiment} $\lambda_i$ to reveal partial information about $\theta_i$.  The commitment assumption is commonly adopted in previous works for information senders \cite{bergemann2019information,kamenica2019bayesian} (including bidders in auctions \cite{shen2019buyer}). 
Such behaviors are also well-justified in real applications. For example, many online advertisers are only willing to reveal user conversion statistics on their websites but not to reveal detailed browsing activities. Such partial information revelation action is usually pre-specified, agreed upon by both parties, and implemented through software. It thus naturally serves the purpose of commitment. 

\subsection{Auctioneer's Inference} 
In an auction with \newmodel{}, the auctioneer receives two pieces of information: the bid profile $\bids=(b_1, b_2, \dots, b_n)$ and  signal profile  $\s= (s_1, s_2, \dots, s_n)$. 
What can the auctioneer infer from $\bids$ and  $\s$? 
Firstly, the auctioneer can infer a distribution of the state profile $\states=(\theta_1, \dots, \theta_n)$. Given $\s$, using the knowledge of $g$ and each bidder's experiment $\lambda_i$, the conditional probability of state profile $\states$ can be derived via Bayes updates as follows,
 \begin{eqnarray}\nonumber
      \Pr(\states| \s) 
      &=& \frac{  \Pr(\s, \states) }{\sum_{\states'} \Pr(\s, \states')} 
    = 
      \frac{   g(\states) \prod_{i\in[n]}\Pr( s_i| \theta_i)  }{  \sum_{\states'}    g(\states')  \prod_{i\in[n]}\Pr( s_i|\theta'_i) }
        \\ \nonumber
          &=& 
          \frac{   g(\states) \prod_{i\in[n]}\Pr(  \theta_i|s_i) \lambda_i(s_i)/g_i(\theta_i)  }{  \sum_{\states'}    g(\states')  \prod_{i\in[n]}\Pr( \theta'_i|s_i) \lambda_i(s_i)/g_i(\theta'_i) } \\\label{eq:updates-state}
      &=& \frac{   g(\states) \prod_{i\in[n]}\Pr(  \theta_i|s_i)/g_i(\theta_i)  }{  \sum_{\states'}    g(\states')  \prod_{i\in[n]}\Pr(\theta_i|s_i)/g_i(\theta'_i) },
 \end{eqnarray}
where: (1) the second equality uses the fact that, conditioned on the information profile $\states$,  each bidder's signal is sampled independently according to its own experiment; (2) equation \eqref{eq:updates-state} cancels out $\lambda_i(s_i)$ on both the denominator and enumerator.
Therefore, similar to the meaning of a signal $s_i$,  we can think of $\s$ directly as a belief over the state profile $\states$, as derived in the above equation \eqref{eq:updates-state}.

Secondly, the auctioneer can determine an expected \emph{value profile} $\v$ of all bidders with the bid profile $\bids$ and the posterior belief of state profile $\theta$ given $\s$, 
\begin{equation}\label{eq:auctioneer-info}
v(\bids; \s) \coloneqq \Ex_{\states \sim \s}[v(\bids;\states)] = \sum_{\states \in \Theta}  \v(\bids ; \states) \Pr( \states|\s). 
\end{equation} 
We similarly denote $v_i( b_i ; \s)$ as the expected valuation of bidder $i$.
If only $\sigma_i$ is fully revealed, we will write bidder $i$'s value as $v_i( b_i ; \sigma_i) := \Ex_{\theta_i\sim \sigma_i} v_i(t_i; \theta_i) $, since $\sigma_i$ is assumed to be a sufficient predictor of $v_i$. 
However, this notational simplification is not valid anymore when bidder $i$ reveals partial information ---  the auctioneer may use the $s_{-i}$ to estimate $\theta_i$, so the input to $v_i$ has to be $\s$ in general. Thus the auctioneer's inference becomes much simpler if bidders reveal full information; this is an advantage of considering truthful mechanisms, as we will do later.   As a convention, here $s_{-i}$ denotes a vector of all signals excluding $s_i$ and  $(\sigma_i, s_{-i})$ denotes the signal profile $\s$ with $s_i$ replaced by $\sigma_i$. The same rule also applies to all other notations such as the type profile $\types$, information profile $\bsigma$, value profile $\v$ or experiment profiles $\blambda$. We can also interpret $\blambda$ as a single aggregated experiment that generates $\s$. We will use $\blambda$ to denote the joint distribution of signal profile $\s$, where the probability of signal profile $\s$ from the aggregated experiment is derived as follows,
\begin{equation}\label{eq:joint-scheme}
    \blambda(\s) = \sum_{\states\in \Theta} \Pr(\s, \states) = \sum_{\states\in \Theta}    g(\states)  \prod_{i\in[n]}\frac{\Pr( \theta_i|s_i) \lambda_i(s_i)}{g_i(\theta_i)}. 
\end{equation} 

Lastly, based on the posterior distribution of $\states$ given $\s$, the auctioneer can infer a value \emph{distribution} $V(\bids;\states)$ such that $V(\bids;\states) = \v(\bids;\states)$ with probability $\Pr(\states|\s)$. Throughout this paper, $v$ will always be used as the deterministic value, whereas the capitalized $V$ is the distribution of values.

\subsection{An Illustrative Example: Ad Auction with CVR Estimation}\label{sec:ad-auction} 

To provide more real-world context, we elaborate on a particular kind of ad auction setup as a concrete instantiation of our model. We also expect the \newmodel{} model to appear in many auctions and pricing problems (e.g., for hotel rooms, concert tickets, etc.). 
Given auctioneers' growing technological and resource advantage from their popular platforms, we believe that more and more auctions on the Internet will feature such information elicitation.

In ad auctions, one primary objective of the advertisers is the conversion value of an Internet user, i.e., after watching a display ads, how much value would the user generate for the advertisers in expectation. Hence, the value of an bidder (i.e., advertiser) is modeled as, $v_i(t_i;\theta_i) = t_i \cdot c_i(\theta_i)$, a product of value per user conversion $t_i$ and the {likelihood} $c_i(\theta_i)$ that the user converts, i.e., the conversion rate (CVR). In practice, $c_i$ can be some carefully-designed function or any black-box machine learning model based on the (latent) user type $\theta_i$~\cite{ContextualSignal}. Each bidder~$i$ has the private knowledge on $t_i$ regarding his margin on user purchases or engagement, as well as $\sigma_i\in \Delta(\Theta_i)$, which reflects a belief of the user type $\theta_i$ based on the features of this Internet user privately collected on the bidder's website. The information assumptions of our model are well justified in such auctions.
The type variable and information variable model different aspects of the bidder's value and thus are naturally independent. Each bidder $i$ usually have sufficient information on the Internet user regarding his own platform; if it is fully revealed, other bidders' information is unlikely to further improve the estimation of the conversion rate to $i$'s platform. So $\sigma_i$ can be assumed as a sufficient predictor of $v_i$. 

The auctioneer (i.e., ad exchange platforms) elicits each bidder's value per conversion $b_i$ and signal $s_i$ as the user data that he chooses to share with the auctioneer. Typically, bidders can choose to reveal information at different levels of granularity, ranging from detailed daily activities for each Internet user to aggregated monthly statistics or customer segmentation, which corresponds to experiments of different informativeness levels (see e.g., ~\cite{CVRtracking, JoinData}). 
The auctioneer will then combine her own data with bidders' data as signal profile $\s$ to infer a belief of user type and estimate a conversion rate $c_i(\s)=\sum_{\theta_i} \Pr(\theta_i|\s) c_i(\theta_i)$ for each bidder $i$. The expected value of bidder $i$ as defined in equation \eqref{eq:auctioneer-info} becomes $v_i(b_i; \s) =  b_i \cdot c_i(\s)$~\cite{ConversionBid, ContextualSignal} (see Section \ref{sec:prelim} for the discussion about connections between feature manipulation and partial information revelation). {In practice, the auctioneer often maps the computed expected value profile to an allocation and payment rule \cite{paes2020competitive,sundararajan2016prediction}.  However, we will show later that this practice may not necessarily be the best for revenue maximization by simply treating bidders' values as the expected value $v_i(b_i; \s)$. The fact that auctioneer now has a distribution of bidder values turns out to be a fundamental difference between our setup and classic auction design problems. 
}

\subsection{Additional Remarks on the Modeling. } \label{sec:model-remark}
Before concluding this  modeling section, we make a few remarks highlighting the subtleties of modeling information revelation in auction design and how our model overcomes these challenges. 
First, we point out that a simpler value model with only the information variable $\sigma_i$ (without the type variable $t_i$) is not only less realistic  but also is not technically interesting.   In this  case,  the auctioneer can simply run the first price auction on the expected value $\v$ given signal profile $\s$. Without the option of misreporting to lower his payment from the expected value, the winner will always pay the  true expected value in the first price auction. Consequently, any information revelation scheme is a weakly dominant strategy, and  the resultant game is trivial and not interesting. Moreover, the Cremer-Mclean mechanism~\cite{creemer1985optimal,cremer1988full} cannot be applied to our problem to  extract full surplus because: (1) that mechanism has only ex-ante incentive constraints whereas we require ex-pose incentive constraints; (2)   their assumption does not necessarily hold due
to  independent $t_i$.

Second, it is worthwhile to discuss the importance of the information variable $\sigma_i$ only as a belief of the state $\theta_i$. This is a realistic assumption as it is generally impossible to observe all uncertainty. Moreover, this is also technically important. Suppose it is publicly known that the bidder directly observes the state realization $\theta_i$ as his information variable. Then, the auctioneer is able to tell whether the bidder is fully revealing his information variable from the inferred belief of the state $\theta_i$. This setup gives rise to  ``threatening'' mechanisms that  simply do not allocate the item unless the auctioneer can fully determine bidder's state. Thus, it is necessary to assume the existence of unobservable uncertainty in the system so that 
even when all bidders reveal full information, the auctioneer still only learns a belief of state --- the auctioneer does not know how much the bidder could know about the state and thus \emph{not} tell whether the bidder is fully revealing or not.


Third, we  discuss the implementation of experiments in practice and its implication to the  mechanism design.
An ``experiment'' or ``data generating process'' generates signals about the hidden random variable $\states$\cite{blackwell1953equivalent,kamenica2019bayesian}. 
It is crucial in our model that the auctioneer directly observes  signal $s_i \sim \lambda_i$  as a distribution over states without the need of any knowledge about the underlying signaling scheme that the bidder may implement to generate the signal.
In practice,  bidder $i$ can implement an experiment by designing software that reveals whether the Internet user visited his websites in the past week or not. Then, depending on the answer to be \texttt{YES} or \texttt{NO}, the auctioneer can update her posterior belief about the CVRs, e.g., through Bayes updates or estimated by a machine learning algorithm using this additional feature. 
Importantly, the signal $s_i$ here --- which is a \texttt{YES} or \texttt{NO} --- does not tell how fine-grained bidder $i$'s true belief on $\theta_i$ is. That is, a \texttt{YES} signal does not tell whether bidder $i$ only has the coarse information about the presence of the Internet user in the past week or he actually has more detailed records of each visit. Consequently, while the auctioneer can interpret any signal, she cannot distinguish whether a bidder revealed all he knows or withheld some information. This subtle point is important in our model --- just as important as in the classic setup to assume that the auctioneer cannot distinguish whether a bidder misreported his type variable or not.  Otherwise, again, the auctioneer could  ``threaten'' the bidder by allocating nothing to the bidder unless the bidder reveals full information. 
\section{The Basics of Mechanism Design with \newmodel{} }\label{sec:basics}

\subsection{Mechanisms for \newmodel{}} 

With \newmodel, the input to a mechanism is a type profile $\bids$ and a signal profile $\s$.\footnote{The later is a key difference from classic auction design, in which the designer has a deterministic bidder value profile given any type report. This difference is due to the use of partial information revelation signaling schemes that lead to uncertainty and randomness. }  Therefore, a mechanism  for any type profile $\bids$ and signal profile $\s$ can be described by a pair $(x, p)$, where $x: T \times \Delta(\Theta) \to \RR^n_{+}$ is the allocation function, $p : T \times \Delta(\Theta) \to \RR^n_{+} $ is the payment function. For single item allocation, $\sum_{i\in[n]} x_i(\bids, \s ) \leq 1$. 
The mechanism under \newmodel{} follows the timing and procedures below:
\begin{enumerate}

\item \textit{Ex ante:} The auctioneer and bidders learn the  \emph{prior distribution} $g(\states), h(\types)$ as well as the value function $\v: T\times \Theta \to \RR^n$. 
\item \textit{Auctioneer commitment:} The auctioneer commits to a mechanism $(x,p)$ which maps any input $(\bids,  \s)$ to the allocation and payment.   
\item \textit{Bidder Commitment:} Each bidder $i$ commits to an experiment $\lambda_i\in \Delta(\Delta(\Theta))$ to reveal information about the state $\theta_i$.
\item \textit{Interim:} As an item arrives, $\sigma_i, t_i$ are realized for each bidder $i$. Bidder $i$ reports his type as bid $b_i$ to the auctioneer. Then a signal $s_i$ is generated by $\tau_i$ and observed by the auctioneer. 
\item \textit{Ex post:} The auctioneer receives the type profile $\bids$, signal profile $\s$ and determines allocation and payment $x(\bids;\s), p(\bids, \s)$ according to the committed mechanism.
\end{enumerate}
Two obvious differences between the above mechanism procedure and classic mechanisms are the following: (1) The classic models do not have Step 3; (2) In Step 4, classic models ask each bidder to report a type variable in the support whereas our model additionally elicits the information signal from the bidder. {We also assume for convenience that each bidder $i$ reports his bid $b_i$ before observing the realization of $s_i$ from $\lambda_i$. This is without loss generality for the design of truthful mechanism for \newmodel{}: if a misreport of $b_i$ at certain $s_i$ is beneficial, then truthfulness must not hold in general, e.g., when $\tau_i$ is a point distribution on such $s_i$. Hence, we can simply think of the bidder's strategy as to pick a $b_i$ based on $t_i$ regardless of the outcome $s_i$ of his experiment. } 

\subsection{Truthful Mechanisms and the Revelation Principle}

Similar to classic mechanism design, we will be interested in the design of \emph{truthful mechanisms} which has a slightly richer meaning here. In particular, our mechanism to provide incentives for truthfulness in both type reporting and information revelation of \newmodel{}. This leads to our study of a truthfulness notion, coined as \emph{Incentive and Information Compatibility} (\infoIC{}). 



We start by deriving bidders' utility functions.  
According to the mechanism for \newmodel{} described above, with the reported type profile $\bids$ and realized signal profile $\s$, given the true type $t_i$, we denote the utility of bidder $i$ as,
\begin{equation} \label{eq:util}
\util (\bids; \s) \coloneqq x_i( \bids, \s ) \cdot v_i( t_i, \s) - p_i(\bids, \s). 
\end{equation} 
where $v_i( t_i, \s) = \Ex_{\states \sim \s}  v_i(b_i, \theta_i)$ is the expected value over randomness of $\states$, conditioned on $\s$. A particularly nice situation is when a bidder fully reveals his information variable $ \sigma_i $. In this case, $i$'s value is fully determined by $\sigma_i, t_i$, and the above utility becomes $ \exputil (\types; \sigma_i, s_{-i})  = x_i( \types; \sigma_i, s_{-i} ) \cdot v_i( t_i ; \sigma_i) - p_i(\types ; \sigma_i, s_{-i})$, where the bidder $i$'s value $v_i( t_i ; \sigma_i)$ is independent of the other bidders' information $s_{-i}$. 

In addition, recall from equation \eqref{eq:joint-scheme}, the joint signaling scheme determines a distribution $\lambda$  over $\s$, we  use the following notation to denote the ultimate expected bidder $i$ utility, under bidding profile $\bids$ and the aggregated experiment $S$: 
\begin{equation} \label{eq:exp-util}
 \exputil (\bids; \blambda) \coloneqq     \sum_{\s} \blambda(\s) \cdot   \util (\bids; \s). 
\end{equation}


To accommodate the randomness of state realization, it is reasonable for each bidder $i$ to focus on their expected utility $\exputil$ under the experiment $\lambda_i$ instead of each realized signal. Hence, the \infoIC{} is to guarantee that reporting the true type and fully revealing experiment $t_i, \tau_i$ maximizes each bidder's expected utility.
In contrast, we will use the utility in equation \eqref{eq:util} instead of the expected utility in equation \eqref{eq:exp-util} for IR constraints, since we want to ensure truthful participation would receive non-negative payoff at each realization of signal.

 
\begin{definition}[Ex-post \infoIC{} and ex-post IR] \label{def:expost} 
A mechanism $(x,p)$ is ex-post \infoIC{} if for every bidder $i$ with any type $b_i, t_i\in T_i$ and experiment $\tau_i \mps \lambda_i$, given other bidders' truthful report $t_{-i}$ and fully revealing experiment $\tau_{-i}$,
\begin{equation}
\exputil (t_i, t_{-i}; \tau_i, \tau_{-i}) \geq \exputil (b_i, t_{-i}; \lambda_i, \tau_{-i})   
\label{eq:inc-any-scheme}
\end{equation}

Similarly, the mechanism is ex-post IR if it satisfies for every bidder $i$, for any {fully revealing signal profile} $(\sigma_i, \sigma_{-i})$ and {truthful type profile} $\types$,
\begin{equation*}
\util (t_i, t_{-i}; \sigma_i, \sigma_{-i}) \geq 0
\end{equation*}

\end{definition}
That is,  participating, reporting true type and revealing full information is a Bayes-Nash equilibrium of the corresponding game in the ex-post stage, i.e., no bidder strictly prefers to partially reveal his information variable or misreporting his type after seeing all other bidders' fully revealing experiments and truthfully reported types. 

\begin{definition}[Dominant-strategy \infoIC{}] 
A mechanism $(x,p)$ is dominant-strategy \infoIC{} if  for every bidder $i$ with any type $b_i, t_i\in T_i$ and experiment $\tau_i \mps \lambda_i$, given any possibly untruthful type profile $b_{-i}$, and possibly partial revealing experiments $\lambda_{-i}$,
\begin{equation*}
\exputil (t_i, b_{-i}; \tau_i, \lambda_{-i}) \geq \exputil (b_i, b_{-i}; \lambda_i, \lambda_{-i})   
\end{equation*}
\end{definition}
That is, reporting true type and revealing full information is a dominant-strategy equilibrium of the corresponding game. The dominant-strategy \infoIC{} is stronger than ex-post \infoIC{} condition. However, dominant-strategy \infoIC{} cannot be expected in general due to a similar reason as in the classic setup -- one bidder’s partial revealing behavior could cause the auctioneer to overcharge a different bidder.  
More details are discussed later in Example \ref{ex:dominant-fail}. We will primarily focus on ex-post \infoIC{}, and thus omit ``ex-post'' unless otherwise specified. 



In this paper, we say a mechanism is \emph{truthful}, if it satisfies ex-post \infoIC{} and IR. This extends the meaning of \emph{truthfulness} in the classical sense, as here the strategically reticent bidders must not only report their true type variables but also fully reveal their information variables. The revelation principle \cite{roughgarden2010algorithmic} also holds under the \newmodel, so the mechanism designer can without loss of generality focus on designing truthful mechanism. 

\begin{theorem} [Revelation Principle]\label{thm:rev-principle}
Under the \newmodel, for any mechanism that implements a Bayes-Nash equilibrium, there always exists another truthful (in ex-post sense), direct mechanism that implements this Bayes-Nash equilibrium. 
\end{theorem} 

The proof of the above revelation principle is standard and thus is deferred to Appendix \ref{append:mech-char}.


\section{The \emm{}}\label{sec:wel-max}
Even though the revelation principle allows us to focus only on truthful mechanisms without loss of generality, designing optimal truthful mechanisms for \newmodel{} turns out still be a difficult task due to the extra signal information that the mechanism should count on.  Since  classic mechanism design  for \oldmodel{}  is already well-understood, to avoid reinventing the wheel, we will build our mechanisms upon previously known ones through reductions. Formally, we will introduce \emph{meta mechanisms}, which take  some base mechanism $\M$ of \oldmodel{} as input and accordingly produces a mechanism (i.e.,  allocation and payment) for our \newmodel{}. Furthermore, we assume the base mechanism $\M$ may be prior dependent, and thus its payment and allocation function takes the input of a reported value profile $\v$ and the prior $F = \{ F_i \}_{i\in [n]}$ of each bidder's value distribution. 
In the current and next section, we will show there does exist a simple meta mechanism that almost perfectly reduces mechanism design with \newmodel{}   to the design in a classic setting. 

We first introduce the \emm{}, which is a natural generalization of mechanisms intended for \oldmodel{}. As shown in the Meta Mechanism \ref{al:expected-meta}, it essentially uses the expected bidder value to determine the allocation and payment according to the input mechanism $\M$. The expected value distribution $F^{\s}$ captures the cumulative probability of the expected value under the signal profile $\s$, and we will see its usage in a concrete example in Section \ref{sec:expected-myerson}.

\begin{algorithm}[h]
    \SetAlgorithmName{Meta Mechanism}{mechanism}{List of Mechanisms}
    \caption{The \emph{Expected} Meta Mechanism }
    \label{al:expected-meta}
    \SetAlgoLined
	\KwIn{The bid profile $\bids$, signal profile $\s$. }
	\KwIn{The mechanism $\M = (x,p)$ designed for \oldmodel{}.}
	\KwOut{The allocation and payment of the expected mechanism over $\M$ for \newmodel{}}
	Compute the expected value profile, 
	$$ v(\bids; \s) =\sum_{\states \in \Theta}  \v(\bids; \states) \Pr( \states|\s), $$
    and the expected value distribution over type for all bidders,
    $$ F^{\s} \coloneqq \left\{ F_i^{\s}(v) = \int_{v_i( b_i; \s) \leq v } h_i(b_i) \mathrm{d} b_i  \right\}_{i\in [n]}. $$    
    Use the expected value profile and distribution to determine the allocation and payment 
    $$
    \bar{x}(\bids;\s) = x( v(\bids; \s), F^{\s}) \qquad  \text{ and } \qquad \bar{p}(\bids;\s) = p( v(\bids; \s), F^{\s}).
    $$
    
    Return the allocation and payment $\bar{x}(\bids;\s), \bar{p}(\bids;\s)$ for \newmodel{}.
\end{algorithm}

For the \emm{}, we assume the base mechanism $\M$ is feasible for any value profile $\v \in \{\v(\bids, \s): \bids\in T, \s \in \Delta(\Theta) \}$ --- most mechanisms are indeed designed to work with any given value profiles. This ensures the allocation and payment functions of $\M$ to be well defined for the expected value profile under any reported type and signal profile.
Moreover, the IC and IR constraints of $\M$ are satisfied, for any bidder $i\in [n]$,~\footnote{Note that the dominant-strategy IC and ex-post IC are equivalent here, since the value of each bidder in our model is independent of other bidders' types. Thus, in this paper, we will simply use IC to refer to the equivalent condition about the base mechanism $\M$.} 
\begin{align}\tag{\text{IC}}
    x_i ( v_i, v_{-i} )  \cdot v_i - p_i (v_i, v_{-i})   
  & \geq   x_i (v'_i, v_{-i} )  \cdot v_i - p_i (v'_i, v_{-i}) 
  &  \forall v'_i \neq v_i, v_{-i},  \\ \tag{\text{IR}}
   x_i (v_i, v_{-i} )  \cdot v_i - p_i (v_i, v_{-i})   
  & \geq   0   
  & \forall v_i, v_{-i}
\end{align}
%
Our next result provides a black-box reduction from \infoIC{} and IR mechanism for \newmodel{} to IC and IR mechanism $\M$ for \oldmodel{}, as long as $\M$ satisfies the following condition.  

\begin{definition}[Proper Mechanism]
A mechanism $\M=(x,p)$ for \oldmodel{} is proper, if any truthful bidder $i$'s utility, $ x_i ( v_i, v_{-i} )  \cdot v_i - p_i (v_i, v_{-i}), $ is convex in $v_i$, for any fixed $v_{-i}$.
\end{definition}

This notion draws inspiration from the proper scoring rule in literature of information elicitation~\cite{savage1971elicitation}. It is well-known that proper scoring rule guarantees that a forecaster's utility is maximized at full revelation of information state. Similarly, the proper mechanism here guarantees that a bidder's utility is maximized at full revelation of his realized value.

\begin{theorem}\label{thm:expected-meta-reduce}
An \emm{} $\bar{\M}$ is ex-post \infoIC{} and IR if and only if its base mechanism $\M$ is proper, IC and IR.
\end{theorem}
\begin{proof}
It is straightforward to verify that at full revelation, the IR condition of $\bar{\M}$ is equivalent to the IR condition of $\M$ by setting the value profile $(v_i, v_{-i}) = \v(\types; \bsigma)$. Hence, the \emm{} $\bar{\M}$ is IR if and only if its base mechanism $\M$ is IR.

It remains to show the \emm{} $\bar{\M}$ is ex-post \infoIC{} if and only if its base mechanism $\M$ is proper and IC.


We start with the "if" direction, and it suffices to show that the following two inequalities holds, respectively, for any bidder $i\in [n]$,
\begin{equation}
 \label{eq:expected-iic-decompose}
 \exputil(t_i,t_{-i}; \tau_{i},\tau_{-i})
  \geq 
  \exputil(t_i,t_{-i}; \lambda_{i},\tau_{-i})
  \geq 
  \exputil(b_i,t_{-i}; \lambda_{i},\tau_{-i}),
  \quad \forall b_i \in T_i, \tau_{i} \mps \lambda_{i}, 
\end{equation}
where $(t_i, t_{-i})$, $(\tau_{i}, \tau_{-i})$ are the true type profile and fully-revealing experiments. 

To ease the proof of the inequalities, we first rewrite the above terms of expected utility by introducing three random variables $V^{\tau}_i, V^{\lambda}_i, V^{\tau}_{-i}$. From any experiment $\tau_i \mps \lambda_i$ and the other bidder's fully-revealing signal $\sigma_{-i}$,  we let $V^{\tau}_i = v_i(t_i; \sigma_{i})$ with probability $\tau_i(\sigma_i)$, $V^{\lambda}_i = v_i(t_i; s_i, \sigma_{-i})$ with probability $\lambda_i(s_i)$ and $V^{\tau}_{-i}=v_{-i}(t_{-i}; \sigma_{-i})$ with $\tau_{-i}(\sigma_{-i})$.~\footnote{For notational convenience, we will simply use $s_i$ to represent the uncertainty on $\theta_i$ under the joint information of both $s_i$ and  $\sigma_{-i}$ in several proofs of \infoIC{} through this paper --- since it is equivalent to think of that the bidder $i$ directly reveals the additional information on $\theta_i$ from its $s_i$, given that the auctioneer's belief of $\theta_i$ is the same.} 
Clearly, $V^{\tau}_i \mps V^{\lambda}_i$, as we can verify that $\forall s_i, \sum_{\sigma_i} \Pr(\sigma_i|s_i) v(t_i;\sigma_i) = \sum_{\sigma_i} \Pr(\sigma_i|s_i)\Pr(\theta_i|\sigma_i) v(t_i;\theta_i) = \Pr(\theta_i|s_i) = v(t_i;\theta_i) = v(t_i;s_i)$ and thus $\Ex[V^{\tau}_i | V^{\lambda}_i] = V^{\lambda}_i$.
As the remaining bidders are fully-revealing their information variable $\sigma_{-i}$, the realization of their value profile $v_{-i}(t_{-i}; \theta_{-i})$ are independent of the bidder's reported type or signal, due to the conditional independence assumption. Recall that $u_i ( v_i, v_{-i} )$ is the truthful bidder $i$'s utility in $\M$.
Hence, the bidder $i$'s expected utility can be equivalently written as, 
$$\exputil (t_i, t_{-i}; \tau_i, \tau_{-i}) = \Ex_{ v_{-i}\sim V_{-i}^{\tau}} \Ex_{ v_i\sim V_i^{\tau}} [ x_i ( v_i, v_{-i} )  \cdot v_i - p_i (v_i, v_{-i}) ],$$ 
$$\exputil (t_i, t_{-i}; \lambda_i, \tau_{-i}) = \Ex_{ v_{-i}\sim V_{-i}^{\tau}} \Ex_{ v_i\sim V_i^{\lambda}} [ x_i ( v_i, v_{-i} )  \cdot v_i - p_i (v_i, v_{-i}) ].$$ 
In addition, we introduce the mapping $\eta: \RR \to \RR$ from the expected value under true type $v_i(t_i; s_i, \sigma_{-i})$ to the expected value under reported type $v_i(b_i; s_i, \sigma_{-i})$. So the expected utility under possibly misreported type $b_i$ can be expressed as,
$$\exputil (b_i, t_{-i}; \lambda_i, \tau_{-i}) = \Ex_{ v_{-i}\sim V_{-i}^{\tau}} \Ex_{ v_i\sim V_i^{\lambda}, v'_i=\eta(v_i)} [ x_i ( v'_i, v_{-i} )  \cdot v_i - p_i (v'_i, v_{-i}) ].$$

By removing their common source of randomness, the inequalities in \eqref{eq:expected-iic-decompose} can be implied by the following two inequalities:
\begin{equation}\label{eq:expected-difference-1}
 \Ex_{v_i\sim V^{\tau}_i} [ x_i ( v_i, v_{-i} )  \cdot v_i - p_i (v_i, v_{-i}) ]
\geq  
\Ex_{v_i\sim V^{\lambda}_i} [ x_i ( v_i, v_{-i} )  \cdot v_i - p_i (v_i, v_{-i}) ], \quad \forall v_{-i},
\end{equation}
\begin{equation}\label{eq:expected-difference-2}
  x_i ( v_i, v_{-i} )  \cdot v_i - p_i (v_i, v_{-i})  \geq   x_i ( v'_i, v_{-i} )  \cdot v_i - p_i (v'_i, v_{-i}), \quad \forall v_{i}, v_{-i}.   
\end{equation} 

Inequality \eqref{eq:expected-difference-1} is due to  Theorem~\ref{thm:blackwell} of \citet{blackwell1953equivalent}, as bidder $i$'s utility is convex under the proper mechanism $\M$, and $V^{\tau}_i \mps V^{\lambda}_i$.
The inequality in \eqref{eq:expected-difference-2} is directly implied by the IC of $\M$.

Lastly, the "only if" direction is much easier to see and we defer formal arguments to Appendix~\ref{append:expected-vickrey}. At a high level, it is to construct the extreme cases where inequalities in \eqref{eq:expected-difference-1} and \eqref{eq:expected-difference-2} becomes necessary, since \eqref{eq:expected-difference-1} is exactly the IC of $\M$ and \eqref{eq:expected-difference-1} holds only if $\M$ is proper by Theorem~\ref{thm:blackwell}.  \qedhere

\end{proof}



\subsection{The Expected Vickrey Auction for Welfare Maximization}

A fundamental result in classic auction setup is that the second price mechanism guarantees dominant-strategy truthfulness and maximizes the welfare at equilibrium. It turns out that both the bidder's utility and the auctioneer's objective in Vickrey auction have some convexity structure that favors the certainty of information. As such, combined with the \emm{}, the Vickrey auction suffices to achieve maximum welfare for the \newmodel{}. 
We defer the proof of Proposition \ref{thm:truthfulness-private-value} and full specification of the Expected-Vickrey Auction to Appendix \ref{append:expected-vickrey}.

\begin{proposition}\label{thm:truthfulness-private-value}
Truthfulness forms a Bayes NE in ex-post stage of Expected-Vickrey Auction for \infobid{}, which achieves maximum welfare.
\end{proposition}

We conclude with two remarks about Proposition \ref{thm:truthfulness-private-value}. First, the value profile $\v$ for \newmodel{} depends on $\states$, that is not observable by any agent. Therefore, it is  information-theoretically impossible to determine the bidder with the highest value, i.e., $\max_{i}v_i(t_i, \theta_i)$. Instead, the maximum achievable welfare, given $\types, \bsigma$, is the largest expected value $\max_{i}v_i(t_i, \sigma_i) = \max_{i}\Ex_{\theta_i\sim \sigma_i }v_i(t_i, \theta_i)$. In another word, The ``maximum welfare'' in Proposition \ref{thm:truthfulness-private-value} is with respect to the expected value  $v_i(t_i;\theta_i )$. Second,  while Proposition \ref{thm:truthfulness-private-value} shows that \emm{} can preserve the optimality of the Vickrey auction for the welfare objective, this can be easily generalized to preserve the optimality of any mechanism (satisfying conditions of Theorem \ref{thm:expected-meta-reduce}) for any objective convex in $\v$, which will be maximized upon full information revelation.

\subsection{Sub-Optimality of Expected-Myerson for Revenue Maximization} 
\label{sec:expected-myerson}
An  interesting corollary of Proposition \ref{thm:truthfulness-private-value} is that truthfulness can also be guaranteed for the second price auction with any reserve price. This is because we can simply view the reserve price as an additional bidder and the same proof would still apply. Consequently, if bidders are symmetric in our model --- all bidders have identical value function, state and type distribution --- then the Expected-Myerson Auction is reduced to the Expected-Vickrey Auction with optimal reserve~\cite{myerson1981optimal}, and thus is also truthful.  Moreover, since the revenue objective is also convex in $\Delta_\Theta$~\cite{fu2012ad}, the maximum revenue can be achieved at bidders' full information revelation.  

\begin{corollary}
When bidders are symmetric,
truthfulness forms a Bayes NE in ex-post stage of Expected-Myerson Auction for \infobid{}, which achieves maximum revenue.
\end{corollary}

Interestingly,  the Expected-Myerson Auction is generally not  revenue-optimal  because the bidder's utility function is no  longer convex, so Theorem \ref{thm:expected-meta-reduce} cannot apply. 
In particular, the following example illustrates how truthfulness may be violated for asymmetric bidders and consequently the expected-Myerson auction cannot be revenue-optimal in general. 


\begin{example}[Non-Truthfulness of the Expected-Myerson Auction]
\label{ex:expected-myerson}
Consider an auction with two bidders. For bidder 1, his type distribution follows $h_1(t)=\frac{2}{t^3}, H_1(t)= 1-\frac{1}{t^2}$ with $t\in[1,\infty)$, and the value function at different state $\theta_1 = \{0, 1\}$ is $v_1(t;0)=t, v_1(t;1)=t^2$ respectively. 
Accordingly, we can compute bidder $1$'s virtual value function at each state as, $\phi_1(t; 0)=\frac{t}{2}, \phi_1(t;1)=0$,\footnote{The formal definition of this notation will be introduced in Section \ref{sec:simulated-myerson}.}
so both value distributions are regular. Meanwhile, we let the type and value function of bidder $2$ be fixed such that his virtual value is always $1$.

Suppose $\theta_1\in \{ 0, 1\}$ is a uniformly random bit, the realization of which is privately observed by bidder $1$.
If bidder $1$ always fully reveals his information, he can only win the auction at $\theta_1=0$, with $t>2$ and paying $2$. So his expected utility is $\frac{1}{2}\int_{2}^{\infty} \frac{2}{t^3} (t - 2) = \frac{1}{96} \approx 0.01$. In contrast, if bidder $1$  reveals no information, his expected value is $v_1(t;\s) = \frac{t^2+t}{2}$ virtual value is $\phi_1(t;\s) = \frac{t^2+t}{2} - \frac{4t^2\left(2t+1\right)}{16t+7}$. He can win the auction with $t>2.11$ and pay $3.2789$, and we can compute his expected utility as $\int_{2.11}^{\infty} \frac{2}{t^3} (\frac{t^2+t}{2} - 3.2789) \approx 0.065$. This means that bidder $1$ would prefer to stay silent in order to get more utility. Hence, the Expected-Myerson Auction is not truthful. 
\end{example}

\section{The \smm{}}\label{sec:rev-max}
In this section, we study a different yet equally natural meta mechanism that instead uses the expectation of outcomes.  As   described in Meta Mechanism \ref{al:simulated-meta}, it utilizes each possible realized value profile and output the expectation of their outcomes. 
Specifically, the allocation (resp. payment) is determined by a linear combination of allocation (resp. payment) at every possible realization of bidder's value profile through \emph{simulation}.
Hence, we will refer to such generalization technique as the \smm{}. 
We make the same assumption that the base mechanism may be prior dependent and thus takes the input of $\v$ and $F$. $F_i^\states(\cdot)$ denotes the bidder $i$'s value distribution for a given state $\theta_i$,
and they can be directly computed from the value function $v_i$ and the type distribution $h_i$. It turns out that this meta mechanism can also convert IC mechanism in \oldmodel{} to \infoIC{} in \newmodel{} under some minor assumptions below. 
\begin{algorithm}[h]
    \SetAlgorithmName{Meta Mechanism}{mechanism}{List of Mechanisms}
    \caption{The \emph{Simulated} Meta Mechanism }
    \label{al:simulated-meta}
    \SetAlgoLined
	\KwIn{The bid profile $\bids$, signal profile $\s$ }
	\KwIn{The mechanism $\M = (x,p)$ designed for \oldmodel{}}
	\KwOut{The allocation and payment of the expected mechanism over $M$ for \newmodel{}}
 
 	For each possible realization of $\states$ given $\s$, compute bidders' value profile $v( \bids ; \states)$ and their corresponding value distribution, 
 	$$
 F^{\states} \coloneqq \left\{ F_i^{\theta}( v) = \int_{\forall b_i, v_i( b_i ; \theta_i) \leq v} h_i(b_i) \mathrm{d} b_i  \right\}_{i\in [n]}
 $$

    Use the probability $\Pr(\states | \s)$ to determine the allocation and payment, 
        \begin{align*}
        \bar{x}(\bids ; \s)  =  \sum_{\states \in \Theta} \Pr(\states | \s) x\left( \v(\bids ; \states), F^{\states} \right)  \qquad  \text{ and } \qquad  
        \bar{p}(\bids ; \s)  =  \sum_{\states \in \Theta} \Pr(\states | \s) p\left( \v(\bids ; \states), F^{\states} \right) 
        \end{align*}
    
    Return the allocation and payment $\bar{x}(\bids;\s), \bar{p}(\bids;\s)$ for \newmodel{}.
\end{algorithm}

\begin{definition}[Information Sufficiency]
A mechanism $\M$ satisfies the information sufficiency condition, if for any bidder $i\in [n]$, for any fixed value profile $v_{-i}$ and type $t_i$,
$$
\Var_{\theta_i \sim \sigma_i}\big[ x_i\big( v_i(t_i; \theta_i), v_{-i} \big)  \big] = 0.
$$
\end{definition}
In words, this condition means that uncertainty of state given the information variable $\bsigma$ will not affect the mechanism's allocation (though it may affect the bidder's values). Namely, the bidders' information variables are \emph{sufficient} for the mechanism $\M$ to determine the allocation.
This is typically satisfied when $\bsigma$ does not contain too much randomness, e.g., in the fully informative case with $\Var_{\theta_i \sim \sigma_i}\big[ v_i(t_i; \theta_i) \big)  \big] = 0$, information sufficiency holds for any $\M$.  
However, generally, the value does not have to be constant:  In the Vickrey auction, this condition holds as long as the highest bidder is determined in regardless of the realized $\states \sim \bsigma$; In Myerson's optimal auction, it is to additionally require the virtual value of the winner to be always positive at any $\states \sim \bsigma$.

\begin{definition}[Regular Values]
The regular value condition is satisfied if for any bidder $i$'s value function $v_i: T_i\times \Theta_i \to \RR$, there exists a total ordering of the states such that, for any $\theta_i, \theta'_i \in \Theta_i$, 
$$
v_i(t_i; \theta_i) \geq v_i(t_i; \theta'_i) \iff 
v_i(t'_i; \theta_i) \geq v_i(t'_i; \theta'_i), \quad \forall t_i, t'_i\in T_i.
$$
\end{definition}
This is a natural condition for the value function in the sense that a good state always leads to a higher value than a bad state. More concretely, in Example \ref{ex:expected-myerson}, the state $\theta_1 = 1$ is better than $\theta_0 = 1$, for any $t\geq 1$. In applications such as ad auctions with CVR estimations, this condition is also satisfied: fixing the value per conversion $t_i$, a high conversion rate always leads to more profit than a low conversion rate.

\begin{theorem}\label{thm:simulated-meta-reduce}
For any IC and IR mechanism $\M$ of \oldmodel{} satisfying the information sufficiency and regular value condition, its \smm{} $\bar{\M}$ is ex-post \infoIC{} and IR.
\end{theorem}
\begin{proof}
Fix any bidder $i$ with the remaining bidders being truthful. For any type profile $t_{i}$, signal profile $s_{i}, \sigma_{-i}$, we construct random variable $V^{s}_i, V^{\sigma}_{-i}$ such that $V^{s}_i = v_i(t_i; \theta_{i})$ with probability $\Pr(\theta| s_i,\sigma_{-i})$ and $V^{\sigma}_{-i} = v_{-i}(t_{-i}; \theta_{-i})$ with probability $\Pr(\theta_{-i}|\sigma_{-i})$.
As the remaining bidders are fully-revealing their information variable $\sigma_{-i}$, the realization of their value profile $v_{-i}(t_{-i}; \theta_{-i})$ are independent of the bidder's reported type or signal, due to the conditional independence assumption. Hence, we can separate the randomness of $\theta_{-i} \sim \sigma_{-i}$ and rewrite the bidder's utility under $\bar{\M} = (\bar{x}, \bar{p})$ as,
\begin{align}\nonumber
\util(t_i,t_{-i}; s_{i},\sigma_{-i})  
& = \bar{x}_i (\types ; s_{i},\sigma_{-i})  \cdot v_i( t_i ; s_{i}) - p_i (\types ; s_{i},\sigma_{-i}) \\ \nonumber
& = \Ex_{v_{-i} \sim V^{\sigma}_{-i}} \Ex_{v_i \sim V^{s}_{i} }[ x_i (v_i, v_{-i})] \cdot \Ex_{v_{-i} \sim V^{\sigma}_{-i}} \Ex_{v_i \sim V^{s}_{i} }[v_i] - \Ex_{v_{-i} \sim V^{\sigma}_{-i}} \Ex_{v_i \sim V^{s}_{i} }[p_i (v_i, v_{-i}) ]  \\ \label{eq:simulated-utility}
& = \Ex_{v_{-i} \sim V^{\sigma}_{-i}} \big[ \Ex_{v_i \sim V^{s}_{i} }[ x_i (v_i, v_{-i}) \cdot v_i - p_i (v_i, v_{-i}) ] -  \Cov_{v_i \sim V^{s}_{i} } [ v_i, x_i(v_i, v_{-i})] \big] 
\end{align}
The information sufficiency condition, $\Var_{v_i \sim V^{\sigma}_{i}}\big[ x_i\big( v_i, v_{-i} \big)  \big] = 0$, implies that, for any random variable $V$, the covariance $\Cov_{v_i \sim V^{\sigma}_{i}, v'\sim V } [ v', x_i(v_i, v_{-i})] = 0$.
Combining with the IR of $\M$, the inequality below holds and thus $\bar{\M}$ satisfies IR,
$$
\util(t_i,t_{-i},;\sigma_{i},\sigma_{-i}) = \Ex_{v_{-i} \sim V^{\sigma}_{-i}} \big[ \Ex_{v_i \sim V^{\sigma}_{i} }[ x_i (v_i, v_{-i}) \cdot v_i - p_i (v_i, v_{-i}) ] -  \Cov_{v_i \sim V^{\sigma}_{i} } [ v_i, x_i(v_i, v_{-i})] \big]  \geq  0. 
$$

To prove $\bar{\M}$ is ex-post \infoIC{}, it suffices to show that the following two inequalities holds, respectively, for any bidder $i\in [n]$,
\begin{equation}
 \label{eq:simulated-iic-decompose}
 \exputil(t_i,t_{-i}; \tau_{i},\tau_{-i})
  \geq 
  \exputil(b_i,t_{-i}; \tau_{i},\tau_{-i})
  \geq 
  \exputil(b_i,t_{-i}; \lambda_{i},\tau_{-i}), 
  \quad \forall b_i \in T_i, \tau_{i} \mps \lambda_{i}, 
\end{equation}
where $(t_i, t_{-i})$, $(\tau_{i}, \tau_{-i})$ are the true type profile and fully-revealing experiments. 

We use two facts to rewrite the above terms of expected utility: (1) The realization of $\sigma_{-i} \sim \lambda_{-i}$ is independent of $\sigma_{i} \sim \lambda_{i}$. (2) Given that $\tau_{i} \mps \lambda_{i}$, there exists some $\pi_i: \Sigma_i \to \Delta(S_i)$ such that $\lambda_i(s_i) = \Pr[ s_i | \pi_i(\sigma_i)] \cdot \tau_i(\sigma_i)$. 
So we can separate the randomness of $\sigma_{-i} \sim \lambda_{-i}$ and $\sigma_{i} \sim \tau_{i}$ in $\exputil$ as,
$$
  \exputil(t_i,t_{-i}; \tau_{i},\tau_{-i}) = \Ex_{\sigma_{-i} \sim \tau_{-i}}  \Ex_{\sigma_{i} \sim \tau_{i}}[\util(t_i,t_{-i}; \sigma_{i},\sigma_{-i})],
$$
$$
  \exputil(b_i,t_{-i}; \tau_{i},\tau_{-i}) = \Ex_{\sigma_{-i} \sim \tau_{-i}}  \Ex_{\sigma_{i} \sim \tau_{i}}[\util(b_i,t_{-i}; \sigma_{i},\sigma_{-i})],
$$
$$
\exputil(b_i,t_{-i}; \lambda_{i},\tau_{-i})
= \Ex_{\sigma_{-i} \sim \tau_{-i}}  \Ex_{s_{i} \sim \lambda_{i}}[\util(b_i,t_{-i}; s_{i},\sigma_{-i})]
= \Ex_{\sigma_{-i} \sim \tau_{-i}} \Ex_{\sigma_{i} \sim \tau_{i}}\Ex_{s_{i} \sim \pi_{i}(\sigma_i)}[\util(b_i,t_{-i}; s_{i},\sigma_{-i})].
$$
By removing the common source of randomness in their expectations, the first inequality in equation \eqref{eq:simulated-iic-decompose} can be implied by, 
\begin{equation}\label{eq:difference-1}
\util(t_i,t_{-i}; \sigma_{i},\sigma_{-i})  \geq \util(b_i,t_{-i}; \sigma_{i},\sigma_{-i}), \quad \forall \sigma_{i}, \sigma_{-i}.
\end{equation} 
Similarly, the second inequality in equation \eqref{eq:simulated-iic-decompose} can be implied by 
\begin{equation}\label{eq:difference-2}
 \util(b_i,t_{-i}; \sigma_{i},\sigma_{-i}) ]  \geq  \Ex_{s_{i} \sim \pi_{i}(\sigma_i)}[\util(b_i,t_{-i}; s_{i},\sigma_{-i})], \quad \forall \sigma_{i}, \sigma_{-i}.   
\end{equation} 

We prove the inequality \eqref{eq:difference-1} and \eqref{eq:difference-2}, respectively by showing the differences of their LHS and RHS are non-negative:
\begin{itemize}[leftmargin=*]
\item To prove inequality \eqref{eq:difference-1}, we have $\forall v_{-i}$,
\begin{align*}
& \util(t_i,t_{-i}; \sigma_{i},\sigma_{-i}) - \util(b_i,t_{-i}; \sigma_{i},\sigma_{-i}) \\
= & \Ex_{v_{-i} \sim V^{\sigma}_{-i}} \big[ \Ex_{v_i \sim V^{\sigma}_{i} }[ x_i (v_i, v_{-i}) \cdot v_i - p_i (v_i, v_{-i}) ] -  \Cov_{v_i \sim V^{\sigma}_{i} } [ v_i, x_i(v_i, v_{-i})] \big] \\
& - \Ex_{v_{-i} \sim V^{\sigma}_{-i}} \big[ \Ex_{v_i \sim V^{\sigma}_{i}, v'_i= \eta(v_i) }[ x_i (v'_i, v_{-i}) \cdot v_i - p_i (v'_i, v_{-i}) ] -  \Cov_{v_i \sim V^{\sigma}_{i}, v'_i= \eta(v_i) } [ v_i, x_i(v'_i, v_{-i})] \big] \\
\geq & \Cov_{v_i \sim V^{\sigma}_{i}, v'_i= \eta(v_i) } [ v_i, x_i(v'_i, v_{-i})] - \Cov_{v_i \sim V^{\sigma}_{i} } [ v_i, x_i(v_i, v_{-i})] 
= 0.
\end{align*}
Here the first equality applies the equation \eqref{eq:simulated-utility} and additionally introduce the mapping $\eta: \RR \to \RR$ to capture the misreporting strategy from the value under true type $v_i(t_i; \theta_i)$ to the value under reported type $v_i(b_i; \theta_i)$. The second inequality uses the IC of the base mechanism $\M$. Finally, the two covariance terms are zero, due to the information sufficiency condition. 

\item To prove inequality \eqref{eq:difference-2}, we have $\forall v_{-i}$,
\begin{align*}
 & \util(b_i,t_{-i}; \sigma_{i},\sigma_{-i}) ] -  \Ex_{s_{i} \sim \pi_{i}(\sigma_i)}[\util(b_i,t_{-i}; s_{i}, \sigma_{-i})]  \\
 = &  \Ex_{\theta_i \sim \sigma_i }[ x_i (v_i(b_i;\theta_i), v_{-i})] \cdot \Ex_{\theta_i \sim \sigma_i }[v_i(t_i;\theta_i)] - \Ex_{\theta_i \sim \sigma_i }[p_i (v_i(b_i;\theta_i), v_{-i}) ]  \\
 & - \Ex_{s_{i} \sim \pi_{i}(\sigma_i)}\big[ \Ex_{\theta_i \sim s_i }[ x_i (v_i(b_i;\theta_i), v_{-i})] \cdot \Ex_{\theta_i \sim s_i }[v_i(t_i;\theta_i)] - \Ex_{\theta_i \sim s_i }[p_i (v_i(b_i;\theta_i), v_{-i}) ]  \big] \\
 = &  \Ex_{s_{i} \sim \pi_{i}(\sigma_i)}\Ex_{\theta_i \sim s_i } [ x_i (v_i(b_i;\theta_i), v_{-i})] \cdot \Ex_{s_{i} \sim \pi_{i}(\sigma_i)}\Ex_{\theta_i \sim s_i } [v_i(t_i;\theta_i)] \\
 & - \Ex_{s_{i} \sim \pi_{i}(\sigma_i)}\big[ \Ex_{\theta_i \sim s_i }[ x_i (v_i(b_i;\theta_i), v_{-i})] \cdot \Ex_{\theta_i \sim s_i }[v_i(t_i;\theta_i)]  \big] \\
 = &  \Cov_{s_{i} \sim \pi_{i}(\sigma_i)} \big[ \Ex_{\theta_i \sim s_i}[ v_i(t_i;\theta_i) ],  \Ex_{\theta_i \sim s_i}[x_i(v_i(b_i;\theta_i), v_{-i})] \big] \geq 0.
\end{align*}
Here the first equality is by definition, the second equality cancels out the payment term, and separates the randomness of $\theta_i \sim \sigma_i$ into  $s_{i} \sim \pi_{i}(\sigma_i)$ and $\theta_i \sim s_i$. The last equality is by the definition of covariance and the term is non-negative by the regular value condition and the monotone allocation of $x$ given by the truthful mechanism $\M$, according to Myerson's characterization of truthful mechanism~\cite{myerson1981optimal}. That is, the state that gives the higher true value always leads to the higher value under reported type and thus higher allocation, i.e., $\forall \theta_i, \theta'_i$, if
$v_i(t_i; \theta_i) \geq v_i(t_i; \theta'_i)$, then  
$v_i(b_i; \theta_i) \geq v_i(b_i; \theta'_i)$ and thus $x_i(v_i(b_i;\theta_i), v_{-i}) \geq x_i(v_i(b_i;\theta'_i), v_{-i})$.\qedhere
\end{itemize}
\end{proof}

\subsection{The Simulated Myerson Auction for Revenue Maximization}
\label{sec:simulated-myerson}
In Section \ref{sec:wel-max}, we found that the \emm{} can be used to achieve maximum welfare but unlikely for maximum revenue. In this section, we present \smm{} as the ideal candidate to leverage Myerson's optimal auction for revenue maximization in \newmodel{}.  This is essentially because the simulation can target its optimization objective at each realization and thus the expected outcome would be a global optimum.  

To simplify our notation, we start by defining the state-dependent value distribution as a function of type $t_i$, $F_i(\cdot; \theta_i) \coloneqq F_i^\theta(v_i(\cdot; \theta_i))$, and accordingly denote its derivative as the density function $f_i(\cdot; \theta_i) $. Let $ \phi_i(\cdot; \theta_i) \coloneqq v_i(\cdot; \theta_i) - \frac{1-F_i(\cdot; \theta_i) }{ f_i(\cdot; \theta_i)} $ denote the state-dependent virtual value function w.r.t. the type and $\phi^{-1}_i(\cdot;\theta_i)$ denote its inverse function. 

To utilize the reduction result in Theorem \ref{thm:simulated-meta-reduce}, we shall first notice that the Myerson's optimal auction as the base mechanism does not necessarily satisfy IC for some state profile $\states\in \Theta$. Going back to the Example \ref{ex:expected-myerson}, the bidder $1$'s virtual value is not monotone w.r.t. his value, so monotone allocation property would violate and thus cannot be IC, according to Myerson's characterization of truthful mechanism~\cite{myerson1981optimal}. 
Therefore, we start with the simple case where we introduce the notion of \emph{regular virtual value condition} given in Definition \ref{def:regularity}  to describe the special value function and type distribution that extend the IC of Myerson's optimal auction to the space of $\Theta$ so the optimization program can be easily solved. 


\begin{definition}[Regular Virtual Value] \label{def:regularity}
The regular virtual value condition is satisfied if the virtual value function is  non-decreasing in the value w.r.t. both type and state.
That is, for any $i\in[n]$,
    $$ 
    v_i(t_i ; \theta_i) \geq v_i(t'_i ; \theta'_i) 
    \iff  \phi_i(t_i ; \theta_i) \geq \phi_i(t'_i ; \theta_i') ,\ \forall t_i, t'_i \in T_i,\ \theta_i, \theta_i' \in \Theta_i . $$
\end{definition}

With the regular virtual value condition, we are able to construct IIC mechanism based on Myerson's optimal auction, as in Mechanism \ref{al:rev-opt}, and we will refer to it the Simulated-Myerson Auction.
Furthermore, we show that with the Simulated-Myerson Auction is revenue optimal. We defer its formal description and the proof of Proposition \ref{prop:simulated-myerson} to Appendix \ref{append:simulated-myerson}.

\begin{proposition}\label{prop:simulated-myerson}
Under the information sufficiency, regular value and regular virtual value condition, truthfulness forms a Bayes NE in ex-post stage of Simulated-Myerson Auction, which achieves the optimal revenue. 
\end{proposition}

The nice results above relies on the \emph{regular virtual value condition}, and next we show that it is possible to relax such condition through the technique, commonly known as the ironing trick~\cite{myerson1981optimal}. Specifically in the proposition below, we show that if the value function is separable as in the Example in Section \ref{sec:ad-auction} where $v_i(t_i ; \states) = t_i c_i(\states)$, then ironing on the type suffices to restore the regular virtual value condition from any distribution (see the formal proof in Appendix \ref{append:simulated-myerson}).

\begin{proposition}\label{prop:ironing}
If a bidder's value function $v_i(t_i; \theta_i)$ is separable in $t_i$ and $\theta_i$, after applying the ironing trick on $t_i$ w.r.t. its CDF distribution $H(t_i)$, any value distribution $F_i^{\theta}(\cdot)$ satisfies the regular virtual value condition. 
\end{proposition}

However, for the case where the value function is not separable, we note that ironing is an non-trivial problem: This is because even if the distribution of $t$ is regular, it is not always true that a state that always leads to higher value would always lead to higher virtual value, i.e., if $v_i(t_i ; \theta_i) > v_i(t_i ; \theta_i'), \forall t_i$, then $\phi_i(t_i ; \theta_i) > \phi_i(t_i ; \theta_i') , \forall t_i$. More concretely, in Example \ref{ex:expected-myerson}, bidder $1$'s type distribution is regular, but we can see that for any $t>1$, while $ v_1(t ; 0) \leq v_1(t;1)$, $\phi_1(t;0) >  \phi_1(t;1)$. Therefore, it remains an open question on the existence of the ironing technique to restore regular virtual value in general setting of this kind.




\section{Information Regulation ---  Restraint for More} \label{sec:regulate}
Till this end, we want to emphasize that all proofs above require the ex-post condition where all other bidders are truthful. So unlike the Vickrey auction or Myerson's optimal auction in classic setting, truthfulness here does not necessarily form a dominant strategy equilibrium.
The following is an example that truthfulness cannot be a dominant strategy in the Expected-Vickrey auction. This is in contrast to the classic results in \oldmodel{} setup in which truthfulness does form a dominant-strategy equilibrium in private value models.  This impossibility of obtaining  dominant-strategy equilibrium in Expected-Vickrey auction is due to the following reason: in our model, the auctioneer can use the information elicited from one bidder to estimate the value of another bidder by leveraging the correlation among bidders' information variables. 
We make this intuition concrete in the following example: 

\begin{example}[Truthfulness is not a dominant strategy in Expected-Vickrey auction.] \label{ex:dominant-fail}
In an auction with three bidders under \infobid{}. Their information variables are $\theta_1, \theta_2, \theta_3 \sim \{0.25, 0.75\}$, which are highly correlated, $\theta_1 = \theta_2 = 1-\theta_3$. Their types are fixed as $t_1=100, t_2=t_3=1$. Their value functions follow the ad auction model as $v_i = t_i \cdot \theta_i$. 

To argue that truthfulness is not a dominant strategy for bidder $1$, let us consider the situation in which   bidder $2$ and $3$ both  reveal no information. We can see that in this case revealing full information is not optimal for   bidder $1$. If  bidder $1$ reveals no information, the estimated state for every bidder is $0.5$ and bidder $1$ would always win the bid and pay the second highest price $0.5$, receiving surplus $50-0.5$. However, if   bidder $1$ reveals full information, with $1/2$ chances he reveals $\theta_1=0.25$ and wins the bid with surplus $0.25*100-0.75$; with $1/2$ chances, he reveals $\theta_1=0.75$  and wins the bid with surplus $0.75*100-0.75$. So for bidder 1, his expected utility $50-0.75$ under fully revelation scheme is worse than his utility if he reveals no information. So full information revealing cannot be a dominant strategy for bidder 1. 
\end{example}

Meanwhile, many bidders in practice could have concerns about the auctioneer using the revealed information of their users or clients to the estimate of other bidders' value, possibly due to various privacy or compliance reasons. As a result, we investigate an \emph{information regulation} setting that the auctioneer is enforced to only use each bidder’s own information to estimate his value. 
In the case of CVR based ad auction described in Section \ref{sec:ad-auction}, such practice is to only use each bidder's own data in the machine learning and prediction process of his CVR function, such that the belief of each bidder's state is marginalized from other bidders' private information. We formalize this setting as \emph{Information Regulation} in Definition \ref{def:info-reg}.

\begin{definition}[Information Regulation]\label{def:info-reg}
In the \emph{Auctioneer Commitment} stage, the auctioneer commits to estimate each bidder $i$'s state with only his own signal $s_i \in \Delta(\Theta_i)$. 
Formally, given signaling profile $(s_i, s_{-i})$, the auctioneer's belief of each bidder $i$'s state is determined as $ \Pr(\theta_i | s_i) $.
\end{definition}

Perhaps surprisingly, regulating the use of information to account for potential privacy concern turns out to be a marvelous practice that would not only preserve revenue or welfare but also make the mechanism more robust ---  one stone three birds! 
Specifically, in Theorem \ref{thm:dominant-meta}, we show that any ex-post \infoIC{} mechanisms becomes dominant strategy \infoIC{} with only an additional step of information regulation (see the formal proof in Appendix~\ref{appendix:info-reg}). 

\begin{theorem}\label{thm:dominant-meta}
If ${\M}$ is ex-post \infoIC{}, ${\M}$ is dominant-strategy \infoIC{} under information regulation.
\end{theorem}

Since the bidders are truthful under information regulation, the allocation and payment (thus, welfare and revenue) would stay the same as if they were to be truthful without information regulation. 
Therefore, with Proposition \ref{thm:truthfulness-private-value} and \ref{prop:simulated-myerson}, we are able to derive the two corollaries below.

    



\begin{corollary}\label{thm:expected-2nd-dominant}
Under information regulation, the Expected-Vickrey Auction is dominant-strategy \infoIC{} and IR, which achieves maximum welfare.
\end{corollary}
\begin{corollary}\label{thm:rev-opt-dominant}
Under information regulation, the Simulated-Myerson Auction under information sufficiency, regular value and regular virtual value condition is dominant-strategy \infoIC{} and IR, which achieves maximum revenue.
\end{corollary}

In retrospect, this setting closely resembles the independent private value setting of classic auctions, as bidders' estimated valuations, without each other's information for inference, are now essentially independent and private. This result informs us of the counter-intuitive benefits of information regulation --- a good practice that auctioneers should follow in the face of \newmodel{}.

\section{Conclusion}
We claim our contribution as to first formalize the framework and explore the landscape of a novel auction design problem with \newmodel{}, which we predict will become increasingly mainstream in this data-driven world. Notably, through the black-box transformation, two different types of simple meta mechanisms leverage the classic auctions to maximize welfare and revenue, respectively. This suggests some fundamental characteristics of the two objectives when information elicitation is involved in the mechanism. In the follow-up work, we intend to study the \newmodel{} in the more general interdependent value setting as well as the multi-item allocation setup. We also want to continue searching for optimal mechanisms under more relaxed assumptions. Meanwhile, in the \newmodel{} model we propose, several open questions emerge. For example, does the correlation of bidders' information variables serves a similar purpose of competition in classic auctions, and how does such competition improves the revenue guarantee of simple auctions (e.g., Mechanism~\ref{al:expected-2nd})? In an online learning setup, can the auctioneer efficiently learn the optimal auction design through interactions with the unknown \newmodel{}?


%
%
%
%
%

\bibliographystyle{ACM-Reference-Format}
\bibliography{refer}

\newpage
\appendix
\section{Missing Details in Mechanism Characterization}
\label{append:mech-char}

To complete our characterization of the solution concepts for \newmodel{}. We also formally define the Bayesian \infoIC{} and interim IR, that is a weaker solution concept than the ex-post \infoIC{} and ex-post IR in our results.

\begin{definition}[Bayesian \infoIC{} and interim IR] 
A mechanism $(x,p)$ is  Bayesian \infoIC{} 
if for every bidder $i$ with type $b_i, t_i\in T_i$ and experiment $\tau_i \mps \lambda_i$, given other bidders' fully revealing experiment $\tau_{-i}$,
\begin{equation*}
\Ex_{t_{-i}\sim h_{-i}} \exputil (t_i, t_{-i}; \tau_i, \tau_{-i}) \geq \Ex_{t_{-i}\sim h_{-i}} \exputil (b_i, t_{-i}; \lambda_i, \tau_{-i})   
\end{equation*}

Similarly, the mechanism is interim IR if it satisfies for every bidder $i$, for any fully revealing signal profile $\sigma_i, \sigma_{-i}$ and {truthful type} $t_i$,
\begin{equation*}
\Ex_{t_{-i}\sim h_{-i}} \util (t_i, t_{-i}; \sigma_i, \sigma_{-i}) \geq 0 
\end{equation*}

\end{definition}

That is,  participating, reporting true type and revealing full information is a Bayes-Nash equilibrium of the corresponding game in the interim stage, where each bidder knows his own private signal but not the others. Note that in this solution concept, it is crucial for each bidder to know the prior distribution of other bidders' states and types in order to estimate his utility.  

Next, we provide the formal proof for the revelation principle.
\begin{theorem}[Restatement of Theorem \ref{thm:rev-principle}]
Under the \newmodel, for any mechanism that implements a Bayes-Nash equilibrium, there always exists another truthful (in ex-post sense), direct mechanism that implements this Bayes-Nash equilibrium. 
\end{theorem} 

\begin{proof}
This proof follows the similar idea of the Revelation Principle in classic setup but with the additional complication of bidder's partial information revelation.
Suppose in the Bayes-Nash equilibrium of $\M$, each bidder $i$ reports the type $b_i$, commits to an experiment $\lambda_i$ and accordingly generate the signal $s_i \sim \lambda_i$. $\M$ thus receives the input $\bids,\s$ and determines an allocation and payment outcome $(x,p)$.
Such equilibrium can be exactly simulated by a truthful, direct mechanism $\bar{\M}$ in following steps:
\begin{enumerate}
    \item $\bar{\M}$ receives the input of truthful type profile $\types$ and information variables $\sigma$.
    \item $\bar{\M}$ constructs the type profile $\bids$ in the equilibrium,
    transforms each information variable $\sigma_i$ into the signal $s_i \sim \pi_i(\sigma_i)$, according to each bidder's committed signaling scheme $\pi_i$ in $\M$, and send them to $\M$ as the input. 
    \item $\bar{\M}$ returns the allocation, payment outcome $(x,p)$ returned by $\M$ at the equilibrium.
\end{enumerate}
As mentioned, since the experiment $\lambda_i \mps \tau_i$, such $\pi_i$ must exists and can be directly constructed through the correlation between the experiment outcomes of $\lambda_i, \tau_i$.
The bidders' truthful behavior is the same as playing the equilibrium strategy in $\M$ and receiving the equilibrium outcome in $\M$. Hence, truthfulness forms the Bayes-Nash equilibrium of $\bar{\M}$.
\end{proof}

\newpage


\section{Missing Details and Proofs of the Expected-Vickrey Auction}
\label{append:expected-vickrey}
In this section, we provide the formal proofs as well as the full descriptions of the Expected-Vickrey Auction, omitted in main paper due to space limit.

\begin{theorem}[Restatement of Theorem \ref{thm:expected-meta-reduce}]
An \emm{} $\bar{\M}$ is ex-post \infoIC{} and IR if and only if its base mechanism $\M$ is proper, IC and IR.
\end{theorem}
\begin{proof}
It is straightforward to verify that at full revelation, the IR condition of $\bar{\M}$ is equivalent to the IR condition of $\M$ by setting the value profile $(v_i, v_{-i}) = \v(\types; \bsigma)$. Hence, the \emm{} $\bar{\M}$ is IR if and only if its base mechanism $\M$ is IR.

It remains to show the \emm{} $\bar{\M}$ is ex-post \infoIC{} if and only if its base mechanism $\M$ is proper and IC.


We start with the "if" direction, and it suffices to show that the following two inequalities holds, respectively, for any bidder $i\in [n]$,
\begin{equation}
 \label{eq:expected-iic-decompose}
 \exputil(t_i,t_{-i}; \tau_{i},\tau_{-i})
  \geq 
  \exputil(t_i,t_{-i}; \lambda_{i},\tau_{-i})
  \geq 
  \exputil(b_i,t_{-i}; \lambda_{i},\tau_{-i}),
  \quad \forall b_i \in T_i, \tau_{i} \mps \lambda_{i}, 
\end{equation}
where $(t_i, t_{-i})$, $(\tau_{i}, \tau_{-i})$ are the true type profile and fully-revealing experiments. 

To ease the proof of the inequalities, we first rewrite the above terms of expected utility by introducing three random variables $V^{\tau}_i, V^{\lambda}_i, V^{\tau}_{-i}$. From any experiment $\tau_i \mps \lambda_i$ and the other bidder's fully-revealing signal $\sigma_{-i}$,  we let $V^{\tau}_i = v_i(t_i; \sigma_{i})$ with probability $\tau_i(\sigma_i)$, $V^{\lambda}_i = v_i(t_i; s_i, \sigma_{-i})$ with probability $\lambda_i(s_i)$ and $V^{\tau}_{-i}=v_{-i}(t_{-i}; \sigma_{-i})$ with $\tau_{-i}(\sigma_{-i})$.~\footnote{For notational convenience, we will simply use $s_i$ to represent the uncertainty on $\theta_i$ under the joint information of both $s_i$ and  $\sigma_{-i}$ in several proofs of \infoIC{} through this paper --- since it is equivalent to think of that the bidder $i$ directly reveals the additional information on $\theta_i$ from its $s_i$, given that the auctioneer's belief of $\theta_i$ is the same.} 
Clearly, $V^{\tau}_i \mps V^{\lambda}_i$, as we can verify that $\forall s_i, \sum_{\sigma_i} \Pr(\sigma_i|s_i) v(t_i;\sigma_i) = \sum_{\sigma_i} \Pr(\sigma_i|s_i)\Pr(\theta_i|\sigma_i) v(t_i;\theta_i) = \Pr(\theta_i|s_i) = v(t_i;\theta_i) = v(t_i;s_i)$ and thus $\Ex[V^{\tau}_i | V^{\lambda}_i] = V^{\lambda}_i$.
As the remaining bidders are fully-revealing their information variable $\sigma_{-i}$, the realization of their value profile $v_{-i}(t_{-i}; \theta_{-i})$ are independent of the bidder's reported type or signal, due to the conditional independence assumption. Recall that $u_i ( v_i, v_{-i} )$ is the truthful bidder $i$'s utility in $\M$.
Hence, the bidder $i$'s expected utility can be equivalently written as, 
$$\exputil (t_i, t_{-i}; \tau_i, \tau_{-i}) = \Ex_{ v_{-i}\sim V_{-i}^{\tau}} \Ex_{ v_i\sim V_i^{\tau}} [ x_i ( v_i, v_{-i} )  \cdot v_i - p_i (v_i, v_{-i}) ],$$ 
$$\exputil (t_i, t_{-i}; \lambda_i, \tau_{-i}) = \Ex_{ v_{-i}\sim V_{-i}^{\tau}} \Ex_{ v_i\sim V_i^{\lambda}} [ x_i ( v_i, v_{-i} )  \cdot v_i - p_i (v_i, v_{-i}) ].$$ 
In addition, we introduce the mapping $\eta: \RR \to \RR$ from the expected value under true type $v_i(t_i; s_i, \sigma_{-i})$ to the expected value under reported type $v_i(b_i; s_i, \sigma_{-i})$. So the expected utility under possibly misreported type $b_i$ can be expressed as,
$$\exputil (b_i, t_{-i}; \lambda_i, \tau_{-i}) = \Ex_{ v_{-i}\sim V_{-i}^{\tau}} \Ex_{ v_i\sim V_i^{\lambda}, v'_i=\eta(v_i)} [ x_i ( v'_i, v_{-i} )  \cdot v_i - p_i (v'_i, v_{-i}) ].$$

By removing their common source of randomness, the inequalities in \eqref{eq:expected-iic-decompose} can be implied by the following two inequalities:
\begin{equation}\label{eq:expected-difference-1}
 \Ex_{v_i\sim V^{\tau}_i} [ x_i ( v_i, v_{-i} )  \cdot v_i - p_i (v_i, v_{-i}) ]
\geq  
\Ex_{v_i\sim V^{\lambda}_i} [ x_i ( v_i, v_{-i} )  \cdot v_i - p_i (v_i, v_{-i}) ], \quad \forall v_{-i},
\end{equation}
\begin{equation}\label{eq:expected-difference-2}
  x_i ( v_i, v_{-i} )  \cdot v_i - p_i (v_i, v_{-i})  \geq   x_i ( v'_i, v_{-i} )  \cdot v_i - p_i (v'_i, v_{-i}), \quad \forall v_{i}, v_{-i}.   
\end{equation} 

Inequality \eqref{eq:expected-difference-1} is due to Blackwell's theorem. That is, bidder $i$'s utility is convex under the proper mechanism $\M$, and $V^{\tau}_i \mps V^{\lambda}_i$.
The inequality in \eqref{eq:expected-difference-2} is directly implied by the IC of $\M$.

Now, the "only if" direction is much easier to see. That is, in the extreme cases when the bidders' value is independent of their states, the value profile only depends on the type profile. Then, ex-post \infoIC{} degenerates to IC,
$$ \exputil(t_i,t_{-i}; \tau_{i},\tau_{-i}) = x_i ( v_i, v_{-i} )  \cdot v_i - p_i (v_i, v_{-i}) \geq x_i ( v'_i, v_{-i} )  \cdot v_i - p_i (v'_i, v_{-i}) = \exputil(b_i,t_{-i}; \lambda_{i},\tau_{-i}),    $$
where $v_i = v_i(t_i;\theta_i), v'_i = v_i(b_i; \theta_i)$ and $v'_{-i} = v_{-i}(t_{-i}; \theta_i)$. Hence, it is necessary for $\M$ to be IC. Meanwhile, in the extreme cases when the bidders' value is independent of their types, the value profile only depends on the signal profile $\s$. Pick any bidder $i$ and let the remaining bidders' value profile be degenerated as $v_{-i}$ such that ex-post \infoIC{} condition degenerates as,
$$ \exputil(t_i,t_{-i}; \tau_{i},\tau_{-i}) 
= \Ex_{v_i \sim V^\tau_i}[ x_i ( v_i, v_{-i} )  \cdot v_i - p_i (v_i, v_{-i}) ] 
\geq \Ex_{v_i \sim V^\lambda_i}[ x_i ( v_i, v_{-i} )  \cdot v_i - p_i (v_i, v_{-i}) ]  
= \exputil(b_i,t_{-i}; \lambda_{i},\tau_{-i}).    $$
By Blackwell's theorem, this inequality holds for any $V^\tau_i\mps  V^\lambda_i$ only if $x_i ( v_i, v_{-i} )  \cdot v_i - p_i (v_i, v_{-i}) $ is convex in $v_i$. Hence, it is also necessary that $\M$ is proper.  

This concludes the proof that $\bar{\M}$ is ex-post IIC and IR.

\end{proof}

\begin{algorithm}[h]
    \SetAlgorithmName{Mechanism}{mechanism}{List of Mechanisms}
    \caption{The \emph{Expected}-Vickrey Auction for \newmodel{}}
    \label{al:expected-2nd}
    \SetAlgoLined
	\KwIn{The bid profile $\bids$, signal profile $\s$}
	\KwOut{The allocation and payment $(x,p)$}
	Compute the expected value  for bidder $i$, 
	$$ v_i(\bids ; \s) =\sum_{\states \in \Theta}  \v_i(\bids ; \states) \Pr( \states|\s) $$

    Set the allocation probability for bidder $i$ as, 
        $$x_i(\bids ; \s) = \begin{cases}
        1 & i = \argmax_{j\in [n]} v_j(\bids ; \s) \\ 
        0 & \textup{otherwise}\\
        \end{cases}
        $$
        
    Set the payment for bidder $i$ as, 
        $$p_i(\bids ; \s) = \begin{cases}
        \max_{j\neq i} v_j(\bids ; \s)   & x_i(V) = 1 \\
        0 & \textup{otherwise}\\
        \end{cases}
        $$
\end{algorithm}

\begin{proposition}[Restatement of Proposition \ref{thm:truthfulness-private-value}]
Truthfulness forms a Bayes NE in ex-post stage of Expected-Vickrey Auction for \infobid{}, which achieves maximum welfare.
\end{proposition}
\begin{proof}
To see that truthfulness forms a Bayes Nash equilibrium in Mechanism \ref{al:expected-2nd}, it suffices to verify that Vickrey auction is a proper mechanism. That is, for any bidder $i$, its utility under value profiles $(v_i, v_{-i})$ is
$x_i ( v_i, v_{-i} )  \cdot v_i - p_i (v_i, v_{-i}) = \max \left\{ 0, v_i - v^* \right\}$, where $v^*=\max_{j\neq i} \{v_j \}$ is a constant independent of $v_{i}$. Hence, the utility function is a maximum function w.r.t. $v_{i}$ and therefore must be convex in $v_i$. By Theorem \ref{thm:expected-meta-reduce}, the Expected-Vickrey Auction is ex-post IIC and IR. 

Since the truthfulness is guaranteed, the mechanism always allocates to the bidder with the highest value at each fully-revealing signal profile $\bsigma$, $ \Ex_{\bsigma \sim \btau} \max_{i} v_i(t_i ; \sigma_i) $, which is the maximum welfare by definition. 
\end{proof}

\newpage

\section{Missing Details and Proofs of the Simulated-Myerson Auction}
\label{append:simulated-myerson}
In this section, we provide the formal proofs as well as the full descriptions of the Simulated-Myerson Auction, omitted in main paper due to space limit.

\begin{algorithm}[h!]
    \SetAlgorithmName{Mechanism}{mechanism}{List of Mechanisms}
    \caption{The \emph{Simulated}-Myerson Auction for \newmodel{}}
    \label{al:rev-opt}
    \SetAlgoLined
	\KwIn{The bid profile $\bids$, signal profile $\s$ }
	\KwOut{The allocation and payment $(x,p)$}
	For each possible realization of $\states$ given $\s$, for each bidder $i$, compute the expected value profile $v( b_i ; \states)$ and the corresponding value distribution over type $\left\{ F_i^{\theta}(\cdot), f_i^{\theta}(\cdot) \right\}_{i\in [n]}$ according to the Meta Mechanism~\ref{al:simulated-meta}, then compute the virtual value function and the corresponding inverse function as, 
$$ \left\{ \phi_i(\ \cdot\ ; \theta_i) = v_i(\ \cdot\ ; \theta_i) - \frac{1-F_i^\theta( \cdot ) }{ f_i^\theta( \cdot )}, 
\phi_i^{-1}(\ \cdot\ ; \theta_i) \right\}_{i\in [n]} 
$$ 
	
    Use the probability $\Pr(\states | \s)$ to compute the allocation for bidder $i$ as, 
    \begin{align*}
        & x_i(\bids ; \s) = \sum_{\states \in \Theta} \Pr(\states | \s) x_i(\bids ; \states) \\
    \textup{where} \quad  & x_i(\bids ; \states) = 
        \begin{cases}
        1 & i = \argmax_{j\in [n]} \left\{ \phi_j( b_j ; \theta_j) \middle| \phi_j( b_j ; \theta_j ) > 0 \right\} \\
        0 & \textup{otherwise}\\
        \end{cases} 
    \end{align*}
        
    Use the probability $\Pr(\states | \s)$ to compute the payment for bidder $i$ as, 
        \begin{align*}
        p_i(\bids ; \s) & =  
        \sum_{\states \in \Theta} \Pr(\states | \s) p_i(\bids ; \states) \\
        \textup{where} \quad  & p_i(\bids ; \states)= 
        \begin{cases}
        v_i(b_i^*; \theta_i)
        & i = \argmax_{j\in [n]} \left\{ \phi_j( b_j ; \theta_j) \middle| \phi_j( b_j ; \theta_j ) > 0 \right\}, \\
        & \text{where } b_i^* = \phi_i^{-1} \left( \max_{j\neq i} \left\{ \phi_j( b_j ; \theta_j) \right\}; \theta_i \right) \\
        0 & \textup{otherwise}\\
        \end{cases} 
        \end{align*}

\end{algorithm}

\begin{proposition}[Restatement of Proposition \ref{prop:simulated-myerson}]
Under the information sufficiency, regular value and regular virtual value condition, truthfulness forms a Bayes NE in ex-post stage of Simulated-Myerson Auction, which achieves the optimal revenue. 
\end{proposition}

\begin{proof}

To maximize the revenue under ex-post \infoIC{} and IR, it is equivalent to solve the following optimization program:
\begin{align*}
  \max &  \Ex_{\bsigma \sim \tau} \Ex_{\types \sim h} \left[\sum_{i \in [n]} p_{i}\left( \types ; \bsigma \right)\right] \\
  \tag{\infoIC{}}
  s.t. 
&  \quad \exputil (t_i, t_{-i}; \tau_i, \tau_{-i}) \geq \exputil (b_i, t_{-i}; \lambda_i, \tau_{-i})  
& \forall i\in [n] \\
 \tag{IR}
&   \quad  \util (t_i, t_{-i}; \sigma_i, \sigma_{-i}) \geq 0
& \forall i\in [n] \\
& \quad  x_i (\types ; \sigma_i, \sigma_{-i}), p_i (\types ; \sigma_i, \sigma_{-i}) \geq 0  & \forall i\in [n]  \\
& \sum_{i\in [n]} x_i (\types ; \sigma_i, \sigma_{-i}) \leq 1 & 
\end{align*}
The objective is set to maximize the auctioneer's revenue in expectation given the truthful \newmodel{} behavior. 
By revelation principle, we can set the input of each payment function as the true type profile $\types \sim h$ and fully-revealing signal profile $\bsigma\sim \tau$, as there must exist an optimal mechanism that satisfies ex-post \infoIC{} and IR.

To argue that Simulated-Myerson auction $\M^*=(x^*, p^*)$, described in Mechanism \ref{al:rev-opt}, is an optimal solution to the above optimization program. We first verify that the Simulated-Myerson auction satisfies each of constraints of the optimization program:

Firstly, Myerson's auction guarantees that $x^*_i (\types ; \states) \geq 0,  p^*_i (\types ; \states) \geq 0, \sum_{i\in [n]}x^*_i (\types ; \states) \geq 0, \forall \types, \states$. This implies that, in the Simulated-Myerson auction, we have $x^*_i (\types ; \sigma_i, \sigma_{-i}) =  \Ex_{\states\sim (\sigma_i, \sigma_{-i})} x^*_i (\types ; \states) \geq 0$, $p^*_i (\types ; \sigma_i, \sigma_{-i}) =  \Ex_{\states\sim (\sigma_i, \sigma_{-i})} p^*_i (\types ; \states) \geq 0$ 
and $$ \sum_{i\in [n]} x^*_i (\types ; \sigma_i, \sigma_{-i}) = \sum_{i\in [n]} \Ex_{\states\sim (\sigma_i, \sigma_{-i})} x^*_i (\types ; \states) =  \Ex_{\states\sim (\sigma_i, \sigma_{-i})} \sum_{i\in [n]}  x^*_i (\types ; \states) \leq 1. $$ 

Secondly, the \infoIC{} and IR constraints of Simulated-Myerson auction are satisfied according to Theorem \ref{thm:simulated-meta-reduce}. In the base mechanism, i.e., Myerson's optimal auction, the allocation is monotonic w.r.t. to the virtual value. Combining with the regular virtual value condition, we know that the allocation is monotonic w.r.t. to the bidder's value. Since the information sufficiency condition is also satisfied, the Mechanism \ref{al:rev-opt} must be \infoIC{} and IR.

Finally, we argue that the maximum revenue is achieved by Mechanism \ref{al:rev-opt}.  As a property of the \smm, the revenue of Mechanism \ref{al:rev-opt} can be equivalently expressed as,
$$
\Ex_{\sigma \sim \tau} \Ex_{\types \sim h} \left[\sum_{i \in [n]} p^*_{i}\left( \types ; \sigma \right)\right]
= \Ex_{\sigma \sim \tau} \Ex_{\types \sim h} \left[\sum_{i \in [n]} \Ex_{\states \sim \sigma} p^*_{i}\left( \types ; \states \right)\right]
=  \Ex_{\states \sim g} \Ex_{\types \sim h} \left[\sum_{i \in [n]} p^*_{i}\left( \types ; \states \right)\right] .
$$
Let $K_0$ denote that the set of all mechanisms $\M=(x,p)$ in \newmodel{} that satisfy both IR and \infoIC{} constraints, $K_1$ denote the set of all \smm{}s $\M'=(x',p')$ with the base mechanism that satisfies IC and IR in \oldmodel{}. Observe that by Myerson's theorem, we have $\Ex_{\states \sim g} \Ex_{\types \sim h} \left[\sum_{i \in [n]} p^*_{i}\left( \types ; \states \right)\right] 
=  \Ex_{\states \sim g} \max_{(x',p')\in K_1}  \Ex_{\types \sim h} \left[\sum_{i \in [n]} p'_{i}\left( \types ; \states \right)\right] $. 
By Theorem \ref{thm:simulated-meta-reduce}, under the information sufficiency condition, any $M'=(p',x')\in K_1$ must also be in $K_0$. Hence, $K_0 \subseteq K_1$. We can derive that 
\begin{align*}
\Ex_{\states \sim g} \max_{(x',p')\in K_1}  \Ex_{\types \sim h} \left[\sum_{i \in [n]} p'_{i}\left( \types ; \states \right)\right] 
\geq & \max_{(x',p')\in K_1} \Ex_{\states \sim g} \Ex_{\types \sim h} \left[\sum_{i \in [n]} p'_{i}\left( \types ; \states \right)\right] 
\geq & \max_{(x,p)\in K_0} \Ex_{\states \sim g} \Ex_{\types \sim h} \left[\sum_{i \in [n]} p_{i}\left( \types ; \states \right)\right],
\end{align*}
where the first inequality follows from Jensen's inequality, and the second inequality is due to the relaxation of domain $K_0 \subseteq K_1$. 
The first term is exactly the revenue achieved by Mechanism \ref{al:rev-opt}. The last term must be no less than the optimization problem's objective, as the optimal mechanism given full knowledge of state can mimic any optimal mechanism with only partial knowledge of state,
$ \max_{(x,p)\in K_0} \Ex_{\states \sim g} \Ex_{\types \sim h} \left[\sum_{i \in [n]} p_{i}\left( \types ; \states \right)\right] 
\geq \max_{(x,p)\in K_0} \Ex_{\sigma \sim \tau} \Ex_{\types \sim h} \left[\sum_{i \in [n]} p_{i}\left( \types ; \sigma \right)\right].$ Therefore, Mechanism \ref{al:rev-opt} indeed achieves the optimal revenue.

\end{proof}

\begin{proposition}[Restatement of Proposition \ref{prop:ironing}]
If a bidder's value function $v_i(t_i; \theta_i)$ is separable in $t_i$ and $\theta_i$, after applying the ironing trick on $t_i$ w.r.t. its CDF distribution $H(t_i)$, any value distribution $F_i^{\theta}(\cdot)$ satisfies the regular virtual value condition. 
\end{proposition}

\begin{proof}
Let the separable value function be $v_i(t_i ; \theta_i) = t_i c_i(\theta_i)$. 
We have
 $\phi_i(t_i ; \theta_i) =  \phi_i(t_i) \cdot c_i(\theta_i) $, where $\phi_i(t_i)$ is essentially the virtual value for $H(t_i)$. This is because of the linearity below,
 \begin{align*}
  \phi_i(t_i ; \theta_i) 
  &=   v_i(t_i ; \theta_i) - \frac{1-F_i(t_i; \theta_i) }{ f_i(t_i ; \theta_i)} \\
  &= v_i(t_i ; \theta_i) - \frac{d v_i(t_i; \theta_i)}{d t_i} \frac{1-H_i(t_i ) }{ h_i(t_i )} \\
  &= t_i c_i(\theta_i) - c_i(\theta_i) \frac{1-H_i(t_i ) }{ h_i(t_i )} \\
  &= \phi_i(t_i) \cdot c_i(\theta_i).
\end{align*}
 So when $t_i$ is regular, the \emph{regular virtual value condition} holds, as the monotonicity is preserved with the additional multiplier $c_i(\theta_i)$. And when $t_i$ is not regular, ironing on $t_i$ makes $\phi_i(t_i)$ regular and thus ${\phi}_i(t_i ; \theta_i) = {\phi}_i(t_i) \cdot c_i(\theta_i) $ strongly regular. 

\end{proof}

\newpage

\section{Proof of Theorem \ref{thm:dominant-meta}} \label{appendix:info-reg}

\begin{theorem}[Restatement of Theorem \ref{thm:dominant-meta}]
If a mechanism ${\M}$ is ex-post \infoIC{}, ${\M}$ is dominant-strategy \infoIC{} under information regulation.
\end{theorem}
\begin{proof}
We prove by showing that for any problem instance $\G$, there exists another problem instance $\G'$ such that if ${\M}$ satisfies the ex-post \infoIC{} constraint in $\G'$, then ${\M}$ satisfies the dominant-strategy \infoIC{} constraint under information regulation in $G$. This would directly imply the Theorem \ref{thm:dominant-meta}, i.e., if ${\M}$ satisfies the ex-post \infoIC{} constraint for any problem instance $\G$, then ${\M}$ must satisfy the dominant-strategy \infoIC{} constraint under information regulation for any problem instance $\G$. 

Therefore, we pick arbitrary problem instance $\G$, where each bidder $i$ privately observes the true type $t_i$ and information variable $\sigma_i\sim \tau_{i}$. By definition, the the dominant-strategy \infoIC{} constraint is satisfied if the following inequality holds, for any $i\in [n]$, for any $b_{-i}\in T_{-i}, \tau_{-i}\mps \lambda_{-i}$,
\begin{equation}
 \label{eq:expected-iic-decompose-dominant}
 \exputil^R(t_i,b_{-i}; \tau_{i},\lambda_{-i})
  \geq 
  \exputil^R(b_i,b_{-i}; \lambda_{i},\lambda_{-i}),
  \quad \forall b_i \in T_i, \tau_{i} \mps \lambda_{i}.
\end{equation}
Under information regulation, the bidder $i$'s value is a function of its own signal $\sigma_i$ or $s_i$, regardless of the signal of other bidders $\sigma_{-i}$ or $s_{-i}$. Hence, the expected utility on LHS and RHS of inequality \eqref{eq:expected-iic-decompose-dominant} can be rewritten as, respectively,
\begin{align} \label{eq:true-full}
\exputil^R(t_i,b_{-i}; \tau_{i},\lambda_{-i}) 
&= \Ex_{s_{-i}\sim \lambda_{-i}}\Ex_{\sigma_i\sim \tau_{i}} [ x_i (t_i, b_{-i} ; \sigma_i, s_{-i})  \cdot v_i( t_i ; \sigma_i) - p_i (t_i, b_{-i} ; \sigma_i, s_{-i})],  \\  \label{eq:false-partial}
\exputil^R(b_i,b_{-i}; \lambda_{i},\lambda_{-i}) 
&= \Ex_{s_{-i}\sim \lambda_{-i}}\Ex_{s_i\sim \lambda_{i}}[ x_i (b_i,b_{-i} ; s_i, s_{-i})  \cdot v_i( t_i ; s_i) - p_i (b_i,b_{-i} ; s_i, s_{-i})].
\end{align}

We now construct a problem instance $\G'$, where the bidder $i$ has the true type $t_i$ and information variable $\sigma_i\sim \tau_{i}$, each bidder $j\neq i$ has the true type $b_j$ and information variable $\s_j\sim \lambda_{j}$.
Then, if $\M$ satisfies ex-post \infoIC{}, we claim that the following inequality holds,
$$
\exputil^R(t_i,b_{-i}; \tau_{i},\lambda_{-i}) = \exputil(t_i,b_{-i}; \tau_{i},\lambda_{-i}) \geq \exputil(b_i,b_{-i}; \lambda_{i},\lambda_{-i}) \geq \exputil^R(b_i,b_{-i}; \lambda_{i},\lambda_{-i}).
$$
Here the first equality is due to the construction of $\G'$ over $\G$. That is, since the bidder $i$ is fully revealing and $\tau_{-i}\mps \lambda_{-i}$, the bidder $i$'s value is $v_i( t_i ; \sigma_i)$ regardless of the other bidder's signal $s_{i} \sim \lambda_{-i}$. The first inequality applies the ex-post \infoIC{} of $\M$. The second inequality uses the fact that aggregated experiment $\lambda_{i},\lambda_{-i}$ w.r.t. state $\theta_i$ must be a mean preserving spread of the experiment $\lambda_{i}$ --- the distribution of belief of $\theta_i$ under information regulation. The ex-post \infoIC{} of $\M$ implies that the bidder $i$'s utility follows the informativeness order of the experiments w.r.t. state $\theta_i$. Therefore, inequality \eqref{eq:expected-iic-decompose-dominant} holds, which concludes the proof.


\end{proof}

\end{document}